\documentclass[a4paper]{article}
\usepackage{amssymb}

\usepackage{amsmath}
\usepackage[english]{babel}
\usepackage[utf8]{inputenc}
\usepackage{algorithm}
\usepackage[noend]{algpseudocode}
\usepackage{natbib}
\usepackage{comment} 
\usepackage{tikz}
\usepackage{caption}
\usepackage{subcaption}
\usetikzlibrary{calc}
\usepackage{aps-bibstyle}
\usepackage{kantlipsum} % for text filler
\def\RR{\mathbb{R}}    % Real field
\def\PP{\mathbb{P}}    % Probability
\def\EE{\mathbb{E}}    % Expectation
\def\proof{\noindent{\em Proof.}~}

\def\eproof{\mbox{\ }\hfill$\square$}
\newtheorem{theorem}{Theorem}
\newtheorem{lemma}{Lemma}
\newtheorem{proposition}{Proposition}
\newtheorem{corollary}{Corollary}
\newtheorem{assumption}{Assumption}
\newtheorem{definition}{Definition}

\newtheorem{example}{Example}
\date{}
\usepackage[symbol*]{footmisc}

\DefineFNsymbolsTM{otherfnsymbols}{%
	\textsection   \mathsection
	\textasteriskcentered *
	\textdaggerdbl \ddagger
	}%

\setfnsymbol{otherfnsymbols}

\begin{document}
	\title{ A Distributionally Robust Random Utility Model }
	
	\author{Emerson Melo\thanks{Department of Economics,
			Indiana University, Bloomington, IN 47408, USA. Email: {\tt emelo@iu.edu}. }
		\and David Müller\thanks{Department of Mathematics, TU Chemnitz, 09126 Chemnitz, Germany. Email: {\tt david.mueller@mathematik.tu-chemnitz.de}} \and Ruben Schlotter\thanks{Hannover Re, Karl-Wiechert-Allee 50, 30625 Hannover, Germany. Email: {\tt ruben.schlotter@hannover-re.com}} }

	\date{\today}
	\maketitle
	
	\pagestyle{myheadings} \thispagestyle{plain} \markboth{ }{ }
	\begin{abstract}  This paper introduces the distributionally robust random utility model (DRO-RUM), which allows the preference shock (unobserved heterogeneity) distribution to be misspecified or unknown. We make three contributions using tools from the literature on robust optimization. First, by exploiting the notion of distributionally robust social surplus function, we show that the DRO-RUM endogenously generates a shock distribution that incorporates a correlation between the utilities of the different alternatives. Second,  we show that the gradient of the distributionally robust social surplus yields the choice probability vector. This result generalizes the celebrated William-Daly-Zachary theorem to environments where the shock distribution is unknown. Third, we show how the DRO-RUM allows us to nonparametrically identify the mean utility vector associated with choice market data. This result extends the demand inversion approach to environments where the shock distribution is unknown or misspecified. We carry out several numerical experiments comparing the performance of the DRO-RUM with the traditional multinomial logit and probit models.
		\end{abstract}
%Thus, all the effort is to provide conditions on the distribution of the random preference shock such that the choice probabilities are consistent with the random utility maximization hypothesis (McFadden $[60]$ )
%Under this additive specification, different distributional assumptions on the random preference term will generate different stochastic choice rules.
	%This paper studies no-regret  learning  in a broad class of environments where a decision maker (DM) has random preferences in a repeated choice setting. 
	%(trust, friendship, cultural differences, etc)
	%
	%\pagestyle{myheadings} \thispagestyle{plain} \markboth{
		%E. MELO }{ PRICING AND LEARNING   IN TRAFFIC NETWORKS }
	\vspace{3ex}
	
	\thispagestyle{empty}
	
	\newcommand{\spacing}[1]{\renewcommand{\baselinestretch}{#1}\large\normalsize}
	\textwidth      5.95in \textheight 600pt
	\spacing{1.1}
	\newpage
%	\section{Introduction}\label{s1} This paper introduces a robust model of information acquisition.
%\cite{Shapiro_1989}
\section{Introduction}\label{Intro}
The random utility model (RUM) introduced by \citet{Marschak1959}, \citet{BlockMarschak1959}, and \citet{BeckerGordonDegroot1963} has become the standard approach to model stochastic choice problems. The fundamental work of  \citet{McFadden1978, mcffadens1, mcf1} makes the RUM  an empirically tractable approach suitable for applications in many areas of applied microeconometric, including labor markets, industrial organization, health economics, transportation, and operations management. In particular, McFadden provides an economic foundation and econometric framework which connects observable to stochastic choice behavior. This latter feature makes the RUM suitable to deal with complex choice environments and welfare analysis (\citet{McFadden2001} and \citet{train_2009}).
\smallskip 

In a RUM  a  decision maker (DM) faces a discrete choice set of alternatives in which each option is associated with a \emph{random} utility. Then the DM chooses a particular option with a probability equal to the event that such alternative yields the highest utility among all available alternatives. Most of the applied literature models the random utility associated with each alternative as the sum of an \emph{observable} and \emph{deterministic} component and a \emph{random preference shock}. Under this additive specification, different distributional assumptions on the random preference shock will generate different stochastic choice rules. Thus, all the effort is to provide conditions on the distribution of the preference shock such that the choice probabilities are consistent with the random utility maximization hypothesis (\citet{mcf1}). 
\smallskip

More importantly, assuming that the shock distribution is known to the analyst, we can estimate the parameters describing the deterministic utility associated with each alternative, carry out counterfactual welfare analysis, and predict future choice behavior. From a modeling standpoint,  this assumption means that the analyst can \emph{correctly} specify the shock distribution that describes the \emph{unobserved} heterogeneity in DM's behavior. 

\smallskip
In this paper, we develop a RUM framework that allows for the possibility that the analyst (or the DM) does not know the true shock distribution. In doing so, we propose a distributional robust framework that relaxes the assumption that the shock distribution is known in advance. In particular, we develop a RUM  framework that allows for misspecification in the shock distribution. By modeling the uncertainty regarding the true distribution, we follow the distributionally robust optimization literature and consider an environment where the analyst has access to a reference distribution  $F$. This distribution corresponds to an approximation of the true statistical law generating the realizations of preference shocks. We refer to $F$ as the nominal distribution. Accordingly,  we model uncertainty distribution in terms of an \emph{uncertainty set}, which consists of all probability distributions that are \emph{close} to  $F$. We rely on the concept of  \emph{statistical} divergences to measure the distance between probability distributions. More precisely, we use the notion of $\phi$-divergences  (\cite{Csiszar1967, LieseVajda1987}). Examples of $\phi$-divergences include the Kullback-Leibler,  Renyi,  and Cressie-Read distances, among many others. Thus, the uncertainty set contains the nominal $F$ and all feasible distributions within a certain radius as measured by the $\phi$-divergence.  
\smallskip

Based on the uncertainty set, we introduce the robust social surplus function, corresponding to the maximum social surplus achievable over all feasible distributions. Like the traditional RUM, the robust social surplus is a convex function that contains all the relevant information to study and understand our distributionally robust RUM (DRO-RUM).

%\textcolor{red}{citation provides error, thats why I temporarly comment it out} 
%Our approach is based in the idea of defining a set of probability distributions that are ``close'' to the nominal distribution $F$, which is assumed to be known.   we develop a distributionally robust  random utility model (DRO-RUM), 

\subsection{Contributions} We make three contributions. First, we show that the analysis of the DRO-RUM corresponds to the study of the properties of a strictly convex finite dimensional stochastic optimization program. This characterization directly implies that the \emph{endogenous}  robust distribution associated with the DRO-RUM introduces correlation between the preference shocks, even when the nominal $F$ may assume independence.

\smallskip

Second, we show that the gradient of the robust social surplus function yields the choice probability vector. The latter result is a nontrivial generalization of the celebrated Williams-Daly-Zachary (WDZ) theorem to environments where the true shock distribution is unknown.   Furthermore, we show that the  DRO-RUM preserves the convex structure of the traditional RUM. In particular, we derive a robust Fenchel duality framework that connects the robust social surplus and its convex conjugate.

\smallskip

In our third contribution, we characterize the empirical content of the DRO-RUM. Formally, we show that for an observed choice probability vector, there exists a \emph{unique}  mean utility vector that rationalizes the observed data in terms of a DRO-RUM. In particular, we show that the mean utility vector corresponds to the gradient of the convex conjugate of the robust social surplus function. The economic content of our result comes from the fact the DRO-RUM  can rationalize observed behavior (a choice probability vector) in terms of a unique mean utility vector, which corresponds to the unique solution of a strictly convex stochastic programming problem.
\smallskip

To conclude our theoretical contributions, we carry out several numerical simulations discussing the properties of our framework. In particular, we compare the choice behavior of the  DRO-RUM with the multinomial logit (MNL) and multinomial probit MNP models. We mainly focus on the impact of the so-called \emph{robustness parameter}, which determines the size of the feasible set, impacting the choice probabilities and the surplus function.
\subsection{Related literature} 
Our paper is related to several strands of literature. First, our paper relates to the literature on RUMs and convex analysis. The closest articles to ours are the works by  \cite{Chiong_Galichon_Shum_2016}, \cite{Galichon_Salanie_2021},  and \cite{Fosgerauetal2021109911}. Similar to us,  these papers exploit the convex structure of the RUM to study the nonparametric identification of the mean utility vector when aggregate market data is available (observed choice probabilities). Our paper and results differ substantially from their work by allowing a more flexible framework regarding distributional assumptions.
 %Furthermore, from a technical standpoint our paper contributes to the literature on RUMs and convex analysis 

\smallskip

Second, our paper relates to the semiparametric choice model (SCM) literature. The work by \cite{Natarajan_et_al_2009}  introduces the SCM in an environment where the true \emph{joint} distribution is unknown, but the analyst has access to the set of marginal distributions associated with each alternative. This particular instance of the SCM is known as the marginal distribution model (MDM). \cite{Mishra_et_al_2014} studies the MDM approach's theoretical and empirical performance. \cite{Mishra_et_al_2012} study a second instance of the SCM, which exploits cross moments constraints. In particular, they assume that the true distribution is unknown but the analyst has access to the true variance-covariance matrix that captures the correlation structure across the set of discrete alternatives. 

\smallskip
%First, our definition of the uncertainty set requires the assumption that only the nominal distribution is known. 
At first glance, our approach is similar to SCM. As discussed by \citet{Feng_et_al_opre.2017.1602}, the latter are generally defined by a supremum over a set of distributions. We adapt to this definition by introducing the DRO-RUM, where the true distribution is unknown, and the analyst, therefore, considers all distributions in an uncertainty set.\\

Despite the similarity between the general SCM and our approach, both frameworks have important differences. First, our approach requires no assumption on the marginals or variance-covariance matrix. Instead, our model only requires knowledge of a nominal distribution. Second, using the concept of $\phi$-divergence enables the researcher to incorporate robustness, where she can control the uncertainty concerning the shock distribution by selecting the \emph{robustness parameter}. Hence, the feasible set is not determined explicitly by fixing some moments or marginal distributions but is rather implicitly constructed by choosing the nominal distribution and the magnitude of the  \emph{robustness parameter}.
Moreover, our approach can generate different models by allowing the choice of several $\phi$-divergence functions and different nominal distributions. Third, we show that the DRO-RUM preserves the convex structure (and duality) of the traditional RUM approach. In particular, we generalize the WDZ  and provide a robust Fenchel duality analysis. Fourth,  we identify the mean utility vector nonparametrically by exploiting our robust convex duality results. This result allows us to rationalize aggregate market data (choice probabilities) in terms of a DRO-RUM. In particular, our identification result corresponds to a robust demand inversion method.
\smallskip
%None of these results are discussed  in the SCM approach. 
Our paper is also related to the literature on robustness in macroeconomics (\cite{Hansen_Sargent_2001,Hansen_Sargent2008}). However, this literature focuses on recursive problems using the Kullback-Leibler distance.   The recent paper by \cite{Christensen_Connault_2022} introduces robustness ideas to analyze the sensitivity of counterfactuals to parametric assumptions about the distribution of latent variables in structural models. Their focus is different from the problem we study in this paper.
\smallskip
\renewcommand*{\thefootnote}{\arabic{footnote}}
\setcounter{footnote}{0}
Finally, our paper is closely related to the literature on distributionally robust optimization. \cite{Shapiro} and \cite{Kuhn_et_al_2019} provide an up-to-date treatment  of the subject.  Applications vary from inventory management to regularization in machine learning.\footnote{In Economics, one of the first papers studying robust optimization problems is \cite{Scarf1958}}However, to our knowledge, this literature has not studied the problem of the distributional robustness of the RUM.
\smallskip

The rest of the paper is organized as follows. Section \ref{sec:RUM_CV} reviews the traditional RUM approach and introduces the problem of robustness. Section \ref{s3:DRO-RUM-Model} presents the DRO-RUM model and discusses its main properties.  The empirical content of  the DRO-RUM approach is discussed in Section \ref{s4:empirical_content}. Section \ref{s5:numerical_experiments} contents several numerical experiments comparing the outcome of the DRO-RUM with respect to MNL and MNP. Section \ref{s6:final_remarks} concludes the paper by providing an overview of possible extensions.
\smallskip

\noindent{\textbf{Notation}}. Throughout the paper we use the following notation and definitions.
Let us denote $\bar{\RR} = \RR \cup\left\{-\infty, +\infty\right\}$ and consider extended real-valued functions 
\[
f: \mathbb{V} \to \bar{\mathbb{R}},
\]
where $\mathbb{V}$ is a finite  dimensional real vector space. Consequently, we denote by $\mathbb{V}^*$ its dual space consisting of all linear functionals. In particular, we often work with   subspaces of $\RR^n$. % with $\R^n$ being the space of $n$-dimensional column vectors $$x =
%\left(x^{(1)}, \dots , x^{(n)}\right)^T,$$
The set defined by 
\[
\mbox{dom} f = \left\{x \in \mathbb{V} : f(x) < +\infty \right\}
\]
is called the {\it (effective) domain} of $f$.  A function is said to be {\it proper} if it takes nowhere the value $-\infty$ and $\mbox{dom} f \neq \emptyset$.
For a proper function $f:\mathbb{V} \to \bar{\mathbb{R}}$ the set $\partial f(x)$ represents its {\it subdifferential} at $x \in \mbox{dom}f$, i.e.
	\[
	\partial f(x) =  \left\{ g \in \mathbb{V}^* : f(y) \ge f(x) + \langle g, y - x\rangle, \quad \mbox{for all} \quad y \in \mathbb{R}^n \right\},
	\] 
	where $g \in \mathbb{V}^* $ is said to be a {\it subgradient}.
	If the subdifferential set is a singleton, i.\,e. the subgradient $g$ is unique, we denote by  $\nabla f(x)$ the {\it gradient} of the function $f$ at $x \in \mbox{int}\left(\mbox{dom}f\right)$.   The {\it convex conjugate} of a proper function  $f:\mathbb{V} \to \bar{\RR}$ is 
	\[
	f^*(g) = \underset{x \in \mathbb{V}}{\sup}\left\{\left\langle x,g \right\rangle - f(x)\right\}, \quad g \in \mathbb{V}^*.
	\]
	 $\EE_F(\cdot)$ denotes the expectation operator with respect to  a distribution $F$.
	
\section{The Random Utility Model}\label{sec:RUM_CV}

Consider a decision maker (DM) making a utility-maximizing discrete choice among alternatives $j\in\mathcal{J}=\left\{0,1,\ldots, J\right\}$.
The utility of option $j$ is
\begin{equation}
	\tilde{u}_{j}=u_{j}+\varepsilon _{j},  \label{random utility}
\end{equation}
where $u=(u_{0},u_1,\ldots ,u_{J})^T$ is deterministic and $\varepsilon=(\varepsilon _{0},\varepsilon_1,\ldots ,\varepsilon _{J})^T$ is a vector of random utility
shocks. The alternative $0$ has the interpretation of an outside option. Following the discrete choice literature, we set $u_0=0$.
\smallskip

Following \cite{McFadden1978,mcf1}, the previous description corresponds to the classic additive random utility model (RUM).  Our presentation of the RUM framework here will emphasize convex-analytic properties.
\begin{assumption}\label{RUM_conditions} 
	The random vector $\varepsilon $ follows a distribution $F$ that  is absolutely continuous with finite means, independent of
	$u$, and fully supported on $\mathbb{R}^{J+1}$.
\end{assumption}

Assumption \ref{RUM_conditions}
leaves the distribution of $\varepsilon$ unspecified, thus allowing
for a wide range of choice probability systems far beyond the often-used logit model. The assumption allows arbitrary correlation between the $%
\varepsilon _{j}$'s may be important in
applications.  As a direct consequence of Assumption \ref{RUM_conditions}, the  DM's choice probabilities correspond to:
\begin{equation*}
	p_{j}(u)\equiv\mathbb{P}\left(u_{j}+\varepsilon
	_{j}=\max_{j^\prime\in\mathcal{J}}\left\{u_{j^\prime}+\varepsilon _{j^\prime}\right\}\right) ,\quad j=0,1,\ldots,J.
\end{equation*}

An important object in the RUM framework is the \emph{surplus function} of the
discrete choice model (so named by \cite{mcf1}). It  is given by
\begin{equation}\label{SSF}
	W\left(u\right)=\mathbb{E}_{F }\left[\max_{ j\in\mathcal{J}}\left\{u_{j}+\varepsilon _{j}\right\}\right].
\end{equation}

Under Assumption \ref{RUM_conditions}, $W$ is convex and differentiable
and the choice probability vector $p(u)$ coincides with the gradient of $W$%
\footnote{The convexity of $W$ follows from the convexity of the max function. Differentiability follows from the absolute continuity of $\varepsilon $. }:
%See \citet{Shi2018}, \citet{Chiong2019}, and \citet{Melo2019} for semiparametric econometric approaches based on these convex-analytic properties of discrete-choice models.}

\begin{equation*}
	\frac{\partial}{\partial u_{k}}W\left(u\right)=p_{k}(u)\,\,\,\mbox{for $k=0,1,\ldots,J$}
\end{equation*}%
or, using vector notation, $p\left(u\right)=\nabla W(u)$. The previous result is the celebrated Williams-Daly-Zachary (henceforth, WDZ) theorem, famous in the discrete choice literature %
(\cite{McFadden1978, mcf1}).
\smallskip

One of the most widely used RUMs is the multinomial logit (MNL) model, which assumes that the entries of $\left(\varepsilon_0, \varepsilon_1, \ldots, \varepsilon_J\right)^T$ follow iid Gumbel distributions with scale parameter $\eta$. Given this assumption,  we can write the social surplus function  in closed form: 
\begin{equation}\label{W_MNL}
	W(u)=\eta\log\left(\sum_{j=0}^Je^{u_j/\eta}\right)  + \eta \gamma,
\end{equation}
where $\gamma$ is the Euler-Mascheroni constant.
It follows from (\ref{W_MNL}) that the WDZ theorem implies that  $p_j(u)$ is given by:
\begin{equation}\label{MNL_Choice}
	\frac{\partial W\left(u\right)}{\partial u_{j}}={e^{u_j/\eta}\over \sum_{l=0}^Je^{u_l/\eta}}\quad\mbox{for $j\in\mathcal{J}.$}
\end{equation}

The MNL model belongs to a broader class of RUM models called generalized extreme value (GEV) models introduced by \citet{gev}. This class of models is defined via a generating function $G:\mathbb{R}^{J+1}_+ \rightarrow \mathbb{R}_+$, which has to satisfy the following properties:  
\begin{itemize}
	\item[(G1)] $G$ is homogeneous of degree $\frac{1}{\eta} > 0$.
	\item[(G2)] $G\left(x_0,x_{1}, \ldots, x_{j}, \ldots, x_{J}\right) \rightarrow \infty$ as $x_{j} \rightarrow \infty$, $j=0, 1, \ldots,J$.
	\item[(G3)] %any subset of indices $\left\{ i_1, \ldots, i_k\right\} \subset \{1,\ldots,n\}$ 
	For the partial derivatives of $G$ w.r.t. $k$ distinct variables it holds: 
	\[
	\frac{\partial^{k+1} G\left(x_{0},\ldots, x_{J}\right)}{\partial x_{j_0}, \partial x_{j_1} \cdots \partial x_{j_k}} \geq 0 \mbox{ if } {k+1} \mbox{ is odd}, \quad 
	\frac{\partial^{k+1} G\left(x_{0},\ldots, x_{J}\right)}{\partial x_{j_0},\partial x_{j_1} \cdots \partial x_{j_k}} \leq 0 \mbox{ if } {k+1} \mbox{ is even}.
	\]
\end{itemize}
\citet{gev, mcf1} show that  a function  $G$ satisfying conditions (G1)-(G3)  implies that the joint distribution of the random  vector $\varepsilon$   corresponds to the  following probability density function:
\[
f_\epsilon\left(y_0,y_{1}, \ldots, y_{J}\right) = \frac{\partial^{J+1} \exp\left(-G\left(e^{-y_{0}},\ldots, e^{-y_{J}}\right)\right)}{\partial y_{0} \cdots \partial y_{J}},
\]
An essential property of the GEV class is that the social surplus function corresponds to \citep{gev}
\[
W(u) = \eta \ln G\left(e^{u}\right) + \eta \gamma,
\]
where $\gamma$ is the Euler-Mascheroni constant.  From the WDZ theorem it follows that the choice probability of the $j$-th alternative corresponds to:
\begin{equation*}\label{eq:gev_choiceprob}
p_{j}(u) =\frac{\partial W(u)}{\partial u_{j}} = \eta \frac{\partial G\left(e^{u}\right)}{\partial e^{u_j}}\cdot \frac{e^{u_{j}}}{G\left(e^{u}\right)}\quad\forall j\in\mathcal{J}.
\end{equation*} 
It is easy to see that  the generating function  
\[
G(e^{u})= \sum_{j =0}^{J} e^{u_j/\eta}=1+\sum_{j =1}^{J} e^{u_j/\eta}
\]
leads to the MNL model.
\smallskip

The main advantage of the GEV class is its flexibility to capture complex patterns correlation across the random variables $\varepsilon_j$'s. Examples of this are the  Nested Logit (NL),  the Paired Combinatorial Logit (PCL), the Ordered GEV (OGEV),  and the Generalized Nested Logit (GNL) model, which are particular instances of the GEV family.

\subsection{A robust framework for   the RUM}\label{ss:spm}
A fundamental assumption in the RUM is that the shock distribution is known to the researcher (and the DM). This means that the distribution of  $\varepsilon$ is correctly specified. Our main goal in this paper is to relax this condition by allowing the distribution of $\varepsilon$ to be unknown. Instead, the distribution of $\varepsilon$ is an argument in an optimization problem that corresponds to the definition of the social surplus function. We formalize this idea by replacing expression  (\ref{SSF}) with the  \emph{ robust   social surplus} function:
\begin{equation}\label{SCM_11}
	W^{RO}(u)=\sup_{G\in \mathcal{M}(F)} \mathbb{E}_{G}\left[\max_{ j\in\mathcal{J}}\left\{u_{j}+\varepsilon _{j}\right\}\right],
\end{equation}
where $\mathcal{M}(F)$ is a set of probability distributions that are close to a predetermined distribution $F$ which satisfies Assumption \ref{RUM_conditions}. This distribution $F$ can be seen as a best guess or prior knowledge of the analyst regarding the joint distribution of error terms. We will refer to $F$ as nominal distribution. In order to be more robust against misspecification the analyst takes into account all possible distributions that are close to the nominal distribution. 
\smallskip
A key aspect of our approach is related to the structure of the set $\mathcal{M}(F)$. In Section \ref{s3:DRO-RUM-Model} we specify the   $\mathcal{M}(F)$ in terms of $\phi$-divergence functions, which enables us to use the notion of statistical divergences between probability distributions (\cite{Csiszar1967}, \cite{LieseVajda1987}, and \cite{Pardo2005}).  Hence,  we will refer to this as the \emph{distributionally robust}-RUM (DRO-RUM).  As we shall see, by doing this we are able to characterize the resulting DRO-RUM surplus function  in terms of a convex finite dimensional optimization program.   This characterization is key in studying the properties of the DRO-RUM approach.
%Hence,  we will refer to this as the \emph{distributionally robust}-RUM (DRO-RUM)
\smallskip

Let $G^\star$ denote the distribution (or a limit of a sequence of distributions) that attains the optimal value in (\ref{SCM_11}). The choice probability for alternative $j$ under this model is given by (provided that it is well defined):
\begin{equation}
	p_j^{RO}(u)=\PP_{G^\star}\left(j=\arg\max_{j^\prime\in\mathcal{J}}\{u_{j^\prime}+\varepsilon_{j^\prime}\}\right)
\end{equation}

From an economic standpoint,  we can interpret the program (\ref{SCM_11}) in two alternative ways. First, the \emph{robust}-RUM  considers a situation where a DM faces preference shocks but has some flexibility concerning the distribution generating those errors. \\ Second, an analyst might not be sure about the distribution of the random vector $\varepsilon$ but might consider a set of possible distributions instead. Thus, it is reasonable for the analyst to assume that the DM is rational and the shock distribution generating the social surplus corresponds to one of the elements in $\mathcal{M}$.

 \subsection{Connection with the semiparametric choice model}\label{Connection_with_SCM} It is worth pointing out that the definition of the RO-RUM is similar to the  \emph{semiparametric choice model} (SCM), which has been recently introduced in the operation research literature (\cite{Natarajan_et_al_2009}). The surplus functions are defined as the supremum over distributions in both model classes. By doing so,  the SCM can capture complex substitution patterns and correlation between the different alternatives in the choice set $\mathcal{J}$. \citet{Feng_et_al_opre.2017.1602} provide a detailed overview of several discrete choice models, where the authors refer to SCM as a supremum over a general set of distributions. Thus, the {\it robust-RUM} could be seen as an instance of a semi-parametric choice model. There are some existing instances of SCM in the literature. In their original paper, \citet{Natarajan_et_al_2009} restrict the feasible set to joint distributions with given information on the marginal distributions. This particular instance of the SCM  is known as the  {\it marginal distribution model} (MDM).\footnote{In MDM, the marginal distributions of the random vector $\varepsilon$ are fixed. Formally, we  write  $\varepsilon_{i} \sim F_i$, where $F_i$ is the marginal  distribution function of the $i$-th error, $i=1, \ldots,J$.  In this case, we define $\mathcal{M}\triangleq \textsc{Mar}=\{F: \varepsilon_i\sim F_i \quad\forall i\in \mathcal{J} \} $. } A second class of  SCMs exploits cross-moment constraints. In particular, \cite{Mishra_et_al_2012} study the \emph{cross-moment model} (CMM), which considers the set $\mathcal{M}$ to be the set of distributions consistent with a \emph{known} variance-covariance matrix.\footnote{Formally, the CMM considers the  set of distributions
	$\mathcal{M}(0,\Sigma)=\{G: \EE_F(\varepsilon)=0,\quad\EE_G(\varepsilon\varepsilon^\top)=\Sigma \}$. In the definition of $\mathcal{M}(0,\Sigma)$,  the variance-covariance matrix $\Sigma$ is assumed to be known. }
\smallskip

Despite the apparent similarity, our approach differs from the existing SCM in several aspects. As we shall see in the rest of the paper, our framework differs from the SCM in the specification of the set of distributions. In particular, in existing SCMs the analyst needs to construct	a feasible set explicitly, for instance, by fixing the marginal distributions. In contrast, in our robust approach, the analyst specifies the feasible set implicitly by determining the nominal distribution $F$ and by upper bounding the distance of other distributions to $F$. Hence,  the DRO-RUM approach does not require knowledge of the marginals of either variance-covariance matrix. In Section \ref{s3:DRO-RUM-Model}, we see that in the DRO-RUM, the researcher controls the distance by selecting the magnitude of a robustness parameter. Thus, our approach follows a rather different principle than existing SCMs.
\smallskip

Additionally,  we show that the DRO-RUM corresponds to the solution of a convex finite dimensional optimization problem. This latter fact allows us to extend the WDZ theorem to environments where the shock distribution is misspecified. Finally, Section \ref{s4:empirical_content} shows how the DRO-RUM enables us to recover the mean utility vector $u$.

 \section{A Distributionally Robust - RUM model}\label{s3:DRO-RUM-Model} 
In this section, we formally introduce the DRO-RUM approach. Following the distributionally robust optimization literature, we consider an environment where the researcher (or the DM) has access to a reference distribution  $F$, which may be an approximation (or estimate) of the true statistical law governing the realizations of $\varepsilon$. We refer to $F$ as the \emph{nominal} distribution. Then, we define a set of probability distributions that are \emph{close} to  $F$. We rely on statistical distances to formalize the notion of distance between probability distributions.

%$\phi$-divergence functions.
%Our approach is based in the idea of defining a set of probability distributions that are ``close'' to the nominal distribution $F$, which is assumed to be known.  
\subsection{$\phi$-divergences} We measure the distance between two probability distributions by the so-called $\phi$-divergence.\\ %(\cite{Csiszar1967}, \cite{LieseVajda1987}, and \cite{Pardo2005}). \\ 
Let $\phi: \mathbb{R} \rightarrow(-\infty,+\infty]$ be a proper closed convex function such that $\mbox{dom} \; \phi$ is an interval with endpoints $\alpha<\beta$, so, $\operatorname{int} \left( \mbox{dom} \; \phi \right) =(\alpha, \beta) $. Since $\phi$ is closed, we have $\lim _{t \rightarrow \alpha_{+}} \phi(t)=\phi(\alpha)$, if $\alpha$ is finite and  $\lim _{t \rightarrow \beta_{-}} \phi(t)=\phi(\beta)$, if $\beta$ is finite.
\smallskip

Throughout the paper we assume that  $\phi$ is nonnegative and  attains its minimum at the point  $1 \in \operatorname{int}  \left(\mbox{dom} \; \phi \right)$, i.e. $\phi(1) = 0$.  The class of such functions is denoted by $\Phi$.
\begin{definition}\label{Phi_Definition}
Given $\phi \in \Phi$, the $\phi$-divergence of the probability measure $G$ with respect to $F$ is
\begin{equation} \label{phi_divergence_definition}
D_{\phi}(G\| F)=\left\{\begin{array}{cl}
	\int_{\RR^{J+1}} \phi\left(\frac{g(\varepsilon)}{f(\varepsilon)}\right) f(\varepsilon)d\varepsilon & \text { if } \mathrm{G}\ll \mathrm{F} \\
	+\infty & \text { otherwise }
\end{array}\right.
\end{equation}
where $f$ and $g$ are the associated densities of $F$ and $G$ respectively.

\end{definition}

%is  when $\varphi(t)=t \log t-t+1$. In this case  the $D_\phi(G\| F)$ recovers the Kullback-Leibler relative entropy. When $\phi(t)=(\sqrt{t}-1)^2$, one obtains the so-called Hellinger distance. \textcolor{blue}{Insert more examples}

To avoid pathological cases, throughout the paper, we assume the following:
\begin{equation}\label{Assumption_phi}
	\phi(0)<\infty, 0 \cdot \phi\left(\frac{0}{0}\right) \equiv 0,0 \cdot \phi\left(\frac{s}{0}\right)=\lim _{\varepsilon \rightarrow 0} \varepsilon  \cdot \phi\left(\frac{s}{\varepsilon}\right)=s \lim _{t \rightarrow \infty} \frac{\phi(t)}{t}, \quad s>0 .
\end{equation}

If the measure $G$ is absolutely continuous w.r.t. $F$, i.\,e. $G\ll F$, the $\phi$-divergence can be conveniently  written as:
\begin{equation}\label{Phi_Def_Exp}
	D_\phi(G\| F)=\EE_{F}\left( \phi(L(\varepsilon))   \right),
\end{equation}
where $L(\varepsilon)\triangleq g(\varepsilon)/f(\varepsilon)$ is the likelihood ratio   between  the densities $g$ and $f$, also known as Radon-Nikodym derivative of the two measures.  Using  the expression (\ref{Phi_Def_Exp}) combined with the convexity of $\phi$, Jensen's inequality implies that 
\begin{equation}\label{Phi_Jensen}
	D_{\phi}(G\| F) \geq \phi\left(\EE_{F}\left(L(\varepsilon)\right)\right)=\phi(1)=0
	\end{equation}
with equality if $G=F$, so that $D_{\phi}(G\| F)$ is a measure of distance of $G$ from $F$.\footnote{We recall that  $1 \in \operatorname{int} \; \mbox{dom} \; \phi$ is the point where $\phi$ attains its minimum 0.}  Furthermore, the $\phi$-divergence functional is convex in both of its arguments. The following proposition summarizes these key properties.
\begin{proposition}\label{phi_properties}The $\phi$-divergence functional (\ref{phi_divergence_definition}) is well-defined and nonnegative. It is equal to zero if and only if $f_1(t)=f_2(t)$ a.e. Furthermore, $D_{\phi}$ is convex on each of its arguments.
	\end{proposition}
\proof The proof that $D_{\phi}(G\| F)$  is well defined and nonnegative follows from \cite[Prop. 1]{BentalTeboulle1987}.  The convexity of $D_{\phi}$ follows from \cite[Prop. 2]{BentalTeboulle1987}.\eproof
\smallskip

In our analysis, a key element will be the convex conjugate of $\phi$. For $\phi \in \Phi$ its conjugate denoted by  $\phi^*$ is:
 \begin{equation}\label{Phi_convex_conjugate}
 	 \phi^*(s)=\sup _{t \in \mathbb{R}}\{s t-\phi(t)\}=\sup _{t \in \operatorname{dom} \phi}\{s t-\phi(t)\}=\sup _{t \in \operatorname{int} \operatorname{dom} \phi}\{s t-\phi(t)\}, 	
 \end{equation}
 where the last equality follows from \cite[Cor. 12.2.2]{Rockafellar1970}. The conjugate $\phi^*$ is a closed proper convex function, with int dom $\phi^*=(a, b)$, where
 $$
 a=\lim _{t \rightarrow-\infty} t^{-1} \phi(t) \in[-\infty,+\infty) ; b=\lim _{t \rightarrow+\infty} t^{-1} \phi(t) \in(-\infty,+\infty] .
 $$
 Moreover, since $\phi$ is convex and closed, we have for its bi-conjugate $\phi^{* *}=\phi$, (\cite{Rockafellar1970}).
 It is worth noting that using the fact that  $1$ is the minimizer of $\phi$ and it is in the interior of its domain, so $\phi'(1) = 0$ holds. In addition,  using the property that $\phi$ is convex and closed,  we have by Fenchel equality 
 $ y =\phi'(x)$ iff $x = \phi^{*'}(y)$. Applying this latter observation to $x=1$ and $y=0$ we obtain $\phi^{\ast\prime}(0)=1$.

\subsection{The DRO-RUM framework}  The main idea is to consider an environment where the analyst (or a DM) does not know the true distribution governing realizations of the shock vector $\varepsilon$. In this environment, the role of $F$ is an approximation or some best guess of the ``true'' unknown distribution. Recognizing this ambiguity or potential misspecification  of the distribution $F$, we  make use of the $\phi$-divergence to  define the \emph{uncertainty set}  $\mathcal{M}_{\phi}(F)$ as:
\begin{equation}\label{Uncertainty_Set}
	\mathcal{M}_{\phi}(F)=\{G\ll F: D_\phi(G||F)\leq \rho\},
\end{equation}

 Formally,  $\mathcal{M}_\phi(F)$  is the set of all probability measures $G$ that are absolutely continuous w.r.t $F$, whose distance from $F$, as measured by the $\phi$-divergence, is at most $\rho$. The hyperparameter  $\rho$ is the radius of $\mathcal{M}_{\phi}(F)$, which reflects how uncertain is the researcher  (or the DM) about the plausibility of $F$ being correct. Let us further elaborate on this interpretation. Following \cite{Hansen_Sargent_2001, Hansen_Sargent2008}, \cite{Shapiro2017}, and \cite{Kuhn_et_al_2019}, we interpret the set  (\ref{Uncertainty_Set}) as an environment in which the analyst (or the DM)  has some best guess $F$ of the true \emph{unknown} probability distribution, but does not fully trust it. For instance, the researcher may consider that the nominal distribution $F$ corresponds to the Gumbel distribution. In this case,  $\mathcal{M}_{\phi}(F)$  accounts for many other probability distributions $G$ to be feasible, where $\rho$ determines the size of the feasible set.  
 \smallskip
 
 Endowed with the set $\mathcal{M}_{\phi}(F)$, we can modify expression \eqref{SCM_11} to obtain a \emph{distributionally robust} surplus function. Thus, the surplus function of   the DRO-RUM corresponds to the following optimization problem: 
\begin{comment}
\begin{equation}\label{SCM_Robust_Formulation}
	W^{DRO}(u)\triangleq \sup_{G\ll F }\left\{\EE_{G}\left(  \max_{j\in\mathcal{J}}{u_j+\varepsilon_j} \right):  D_\phi(G||F)\leq \rho\right\}
	\end{equation}
where the hyperparameter  $\rho>0$ measures the distance between $F$ and $G$. Here 

%where  $G\ll F$ means that $G$ is absolutely continuous with respect to $F$ and $G$ and 
%\begin{equation}\label{Phi_divergence}
%	D_\phi(G||F)\triangleq \int \phi\left({dG\over dF}\right)dF
%\end{equation}
is the $\phi$-divergence between the distributions $F$ and $G$.  The notion of $\phi$-divergence was introduced  \cite{Csiszar1967}. For an in-depth discussion of $\phi$-divergences we refer the reader to the book by \cite{LieseVajda1987}.

In the definition of $D_\phi(G||F)$ the function $\phi:\RR\longrightarrow\bar{\RR}_+=\mathbb{R}_{+} \cup\{\infty\}$ is assumed to be  convex and lower semicontinuous satisfying $\phi(1)=0$ and $\phi(t)=+\infty$ for any $t<0$.

 The $D_\phi(G||F)$ is the statistical distance between $F$ and $G$. \textcolor{red}{Maybe we should recall the definition of the $\phi(G||F)$ divergence?}
\end{comment}
\begin{equation}\label{Robust_SCM}
W^{DRO}(u)=	\sup_{G\in \mathcal{M}_{\phi}(\mathcal{F})}\left\{ \mathbb{E}_{G}\left[\max_{ j\in\mathcal{J}}\left\{u_{j}+\varepsilon _{j}\right\}\right]\right\}
\end{equation}

 %=\int_{\RR^J}\phi\left({g(\varepsilon)\over f(\varepsilon)}\right)f(\varepsilon)d\varepsilon
%Following the robust optimization literature we denote $F$ as the nominal distribution.	
%Intuitively, the nominal distribution $F$ is an approximate model describing the law that governs shock realizations and  $\mathcal{M}_{\phi}(F)$ defines and uncertainty (ambiguity) set that considers all distributions $G$ \textcolor{brown}{whose $\phi$-divergence from  the nominal distribution is at most $\rho$.}
Some remarks are in order. First,  a fundamental aspect of program (\ref{Robust_SCM}) is the role 
of the parameter  $\rho$ which controls the size of  $\mathcal{M}_{\phi}(F)$. Because of this, we can interpret  $\rho$ as an  \emph{index of robustness}. More precisely, when $\rho = 0$ we get  $\mathcal{M}_{\phi}(F)=\{F\}$, which means that we recover the RUM under the distribution $F$.\footnote{We note that  $\rho=0$ implies that  $D_{\phi}(G\| F)=0$. Then by Proposition \ref{phi_properties}, we know that this latter equality holds if and only if $F=G$. }  On the other hand, when  $\rho\longrightarrow\infty$ the uncertainty set $\mathcal{M}_\phi(F)$ admits a much larger set of possible distributions, including those that may not satisfy  Assumption \ref{RUM_conditions}.\footnote{To see this, we note that when $\rho\longrightarrow \infty$ the $\phi$-divergence  is unbounded. This latter fact implies that the set $\mathcal{M}_{\phi}(F)$ consists of all distributions which are absolute continuous w.r.t. to $F$. As $F$ is fully supported, this only implies that the distributions in $\mathcal{M}_{\phi}(F)$   must be continuous but certainly not fully supported on $\RR^{J+1}$. In fact, $\mathcal{M}_{\phi}(F)$ may consist of distributions that are absolutely continuous w.r.t Lebesgue measure but without finite means. For instance, the Pareto distribution with shape parameter $\alpha=1$ is absolutely continuous but fails to have a finite mean. } 
The DRO-RUM aims to set $\rho$ to reflect the perceived uncertainty that the researcher  (or a DM) experiences about the distributional assumption for $\varepsilon$. 
%In this case, our definition of  $W^{DRO}(u)$ recovers the SCM (see discussion in Section \ref{Connection_with_SCM}).
\smallskip

The following lemma establishes some elementary properties of $W^{DRO}(u)$.

	\begin{lemma}\label{lem:surplus}
		For the DRO-RUM the surplus function $W^{DRO}(u)$ satisfies: 
		\begin{itemize}
			\item[(i)] $W^{DRO}(u + c \cdot e) = W^{DRO}(u) + c$ for all $c \in \RR, u \in \RR^J$.
			\item[(ii)] $W^{DRO}(u) \geq W^{DRO}(v)$ for all $u,v \in \mathbb{R}^J$ with $u \geq v$.
			\item[(iii)] $W^{DRO}(u) \geq \displaystyle\max_{j\in \mathcal{J}} u_i + \min_{j\in\mathcal{J}}\mathbb{E}_F\left[\varepsilon_i\right]$.
			\item[(iv)] $W^{DRO}(u)$ is convex in $u$.
		\end{itemize}
	\end{lemma}
	\begin{proof}
		\begin{itemize}
			\item [(i)] The definition provides
			\[
			W^{DRO}(u + c\cdot e) = \sup_{G\in \mathcal{M}_{\phi}(F)}\left\{ \mathbb{E}_{G}\left[\max_{ j\in \mathcal{J}}\left\{u_{j}+\varepsilon _{j} + c \right\}\right] \right\}
			\]
			Due to the linearity of the expectation, it holds
			\[
			c + \sup_{G\in \mathcal{M}_{\phi}(F)} \mathbb{E}_{G}\left[\max_{ j\in \mathcal{J}}\left\{u_{j}+\varepsilon _{j}  \right\}\right] = c + W^{DRO}(u).
			\]
			\item[(ii)] Take any $u,v \in \mathbb{R}^J$ with $u \geq v$. First we note that for any arbitrary feasible distribution $G \in \mathcal{M}_{\phi}(F)$ it holds   
			\[
			W^{DRO}(u) \geq  \mathbb{E}_{G}\left[\max_{ j\in \mathcal{J}}\left\{u_{j}+\varepsilon _{j}  \right\}\right] \overset{(*)}{\geq} \mathbb{E}_{G}\left[\max_{ j\in\mathcal{ J}}\left\{v_{j}+\varepsilon _{j}  \right\}\right],
			\]
			where $(*)$ holds due to the monotonicity of the expectation. Taking the supremum on the right-hand side, we conclude that $	W^{DRO}(u) \geq 	W^{DRO}(v)$.
			\item[(iii)] We deduce that for any $i\in \mathcal{J}$
			\[
			W^{DRO}(u) \geq  \mathbb{E}_{F}\left[\max_{ j\in \mathcal{J}}\left\{u_{j}+\varepsilon _{j}  \right\}\right] \geq \mathbb{E}_{F}\left[u_{i}+\varepsilon _{i}  \right] \geq u_i + \min_{j\in\mathcal{J}} \mathbb{E}_{F}\left[\varepsilon_j\right],
			\]
			which is finite due to Assumption \ref{RUM_conditions}.
	\item[(iv)]	 Let $\alpha\in [0,1]$ and let $u$ and $v$ two deterministic utility vectors.  For a fix distribution $G\in \mathcal{M}_\phi(F)$, Then, due to the convexity of the $\max\{\cdot\}$ operator,
	$$W^{DRO}(\alpha u+(1-\alpha)v)\leq\alpha \EE_{G}\left(\max_{j\in \mathcal{J}}\{u_j+\varepsilon_j\}\right)+(1-\alpha)\EE_{G}\left(\max_{j\in \mathcal{J}}\{v_j+\varepsilon_j\}\right).$$
In the right-hand side, taking the supremum with respect to $G$ over $\mathcal{M}_{\phi}(F)$, we get 
$$W^{DRO}(\alpha u+(1-\alpha)v)\leq \alpha W^{DRO}(u)+(1-\alpha)W^{DRO}(v).$$
Then the convexity of $W^{DRO}(u)$ follows.
	
		\end{itemize}
		
				\end{proof}
	
 %Thus $\rho$ can be seen as a risk-level parameter. 
The following result characterizes $W^{DRO}(u)$.
\begin{proposition}\label{RSCM_Characterization} Let Assumption \ref{RUM_conditions} hold and define the random variable $H(u,\varepsilon)\triangleq \max_{j\in \mathcal{J}}\{u_j+\varepsilon_j\}$. Then,  problem  \eqref{Robust_SCM} is  equivalent to solving the following finite-dimensional convex program:
	\begin{equation}\label{RSCM_Duality}
	W^{DRO}(u)=\inf _{\lambda \geq 0, \mu\in \RR}\left\{\lambda \rho+\mu+\lambda\mathbb{E}_{F}\left[ \phi^{*}\left({H(u,\varepsilon)-\mu\over\lambda}\right)
	\right]\right\},
	\end{equation}
where $\lambda$ is the Lagrange multiplier associated to the uncertainty set $\mathcal{M}_{\phi}(F)$ and $\mu$ the multiplier associated to $G$ being a probabality measure.  Furthermore, the program (\ref{RSCM_Duality}) is convex in $\mu$ and $\lambda$.
	%where $Z\triangleq \max_{j\in J}\{u_j-\alpha(p)+\varepsilon_j\}$.
\end{proposition}
%\textcolor{red}{I think $\lambda c$ should be $\lambda \rho$? Yes, corrected}
%\proof The proof of this result follows from \cite{BenTaletal_2013} and \cite{Shapiro2017} [Need to add the details for completeness].\eproof

%\textcolor{red}{Can we say something about the differentiability w.r.t. to u? Refs might be interested in the derivation of the choice probabilities of the new SPM.}  \textcolor{blue}{I think is differentiable..at least  subdifferentiable....}\\

%\textcolor{blue}{EM: maybe a bit silly comment, but can we use the characterization in Proposition \ref{RSCM_Characterization} to show that $W^{rscm}(u)$ is choice welfare function as in \cite[Def. 1]{Feng_et_al_opre.2017.1602}?}\\
\proof This result follows from a direct application of \cite[Prop. 7.9]{Shapiro_et_al_2013}. For completeness, we provide the details of the argument. First, we note that for a fixed utility vector $u$ and using  the likelihood ratio $L(\varepsilon)\triangleq d G(\varepsilon) / d F(\varepsilon)$, the DRO-RUM  in  (\ref{Robust_SCM}) can be expressed as: 
\begin{comment}
\begin{equation}\label{RSCM_Program1}
	\sup_{G\ll F} \left.\left\{\mathbb{E}_{G}[H(u,\varepsilon )]: D_{\phi}\left(G \| F\right) \leq \rho\right\}\right\}
\end{equation}
where  the $\phi$-divergence between $F$ and $G$ is given by:
$$D_{\phi}\left(G \| F\right)\triangleq\int \phi\left(\frac{d G}{d F}\right) d \mu,$$
where $\phi: \mathbb{R} \rightarrow \overline{\mathbb{R}}_{+}=\mathbb{R}_{+} \cup\{\infty\}$ is a convex lower semicontinuous function satisfying $\phi(1)=0$ and $\phi(t)=+\infty$ for any $t<0$.
\smallskip
\end{comment}

\begin{eqnarray}
W^{DRO}(u)&=&\sup_{G}\{ \EE_{G}(H(u,\varepsilon)): G\in \mathcal{M}_{\phi}(F) \}\nonumber\\
	&=& \sup_{L\geq 0}\left\{\mathbb{E}_{F}[L(\varepsilon) H(u;\varepsilon )] \mid \mathbb{E}_{F}[\phi(L(\varepsilon))] \leq \rho, \mathbb{E}_{F}[L(\varepsilon)]=1\right\}\label{RSCM_Program_11}
\end{eqnarray}
where the supremum is over a set of measurable functions. 
\smallskip

 The Lagrangian of problem (\ref{RSCM_Program_11}) is :
 \begin{equation}\label{RSCM_Lagrangian_11}
 	\mathcal{L}(L, \lambda, \mu)=\int_{\RR^{J+1}}[L(\varepsilon)H(u,\varepsilon) -\lambda \phi(L(\varepsilon))-\mu L(\varepsilon)] d F(\varepsilon)+\lambda \rho+\mu .
 \end{equation}

The Lagrangian dual of problem (\ref{RSCM_Lagrangian_11}) is the problem
\begin{equation}\label{RSCM_Lagrangian_22}
	\inf _{\lambda \geq 0, \mu\in\RR} \sup _{L\geq 0} \mathcal{L}(L, \lambda, \mu)
	\end{equation}

Since Slater condition holds for problem (\ref{RSCM_Lagrangian_11})\footnote{For instance, we can take $L(\varepsilon)= 1$ for all $\varepsilon\in \RR^{J+1}$. }, there is no duality gap between (\ref{RSCM_Lagrangian_11}) and its dual problem (\ref{RSCM_Lagrangian_22}). Moreover, the dual problem has a nonempty and bounded set of optimal solutions.
% In this case the $W^{DRO}(u)\longrightarrow W^{SCM}(u)$. In the DRO-RUM, the goal is to set $\rho$ to reflect the perceived uncertainty that the  DM (or the researcher) experiences in the distributional assumption for $\varepsilon$. %\textcolor{brown}{Thus, smaller values of of $\rho$ are in general favorable.}
\smallskip

By the interchangeability principle (\cite[Thm. 3A]{Rockafellar_1976}), the maximum in (\ref{RSCM_Lagrangian_22}) can be taken inside the integral, that is
$$
\begin{gathered}
	\sup _{L \geq 0} \int_{\RR^{J+1}}[L(\varepsilon)H(u,\varepsilon) -\mu L(\varepsilon)-\lambda \phi(L(\varepsilon))] d F(\varepsilon) \\
	=\int_{\RR^{J+1}}{\sup _{t\geq 0}\{t(H(u,\varepsilon)-\mu)-\lambda \phi(t)\}} d F(\varepsilon),
\end{gathered}
$$
Noting that 
$(\lambda\phi)^*(H(u,\varepsilon)-\mu)=\sup _{t\geq 0}\{t(H(u,\varepsilon)-\mu)-\lambda \phi(t)\}$, then it follows that 
\begin{equation}\label{RSCM_Lagrangian_3}
W^{DRO}(u)=\inf _{\lambda \geq 0, \mu\in \RR}\left\{\lambda \rho+\mu+\mathbb{E}_{F}\left[(\lambda \phi)^{*}(H(u,\varepsilon)-\mu)\right]\right\}.
\end{equation}

To show the convexity with respect to $\lambda$ and $\mu$ we note that it suffices in  (\ref{RSCM_Lagrangian_22}) and (\ref{RSCM_Lagrangian_3}) to take the  $\inf$ with respect to $\lambda>0$ rather than $\lambda \geq 0$, and that $(\lambda \phi)^*(y)=\lambda \phi^*(y / \lambda)$ for $\lambda>0$. Therefore $W^{DRO}(u)$ is given by the optimal value of the following problem:
\begin{equation}\label{RSCM_Lagrangian_4}
\inf_{\lambda>0, \mu\in \RR}\left\{\lambda \rho+\mu+\lambda \mathbb{E}_F\left[\phi^*((H(u,\varepsilon)-\mu) / \lambda)\right]\right\}
\end{equation}
Note that $\phi^*(\cdot)$ is convex. Hence, $\lambda \phi^*(y / \lambda)$ is jointly  convex  in $y$ and $\lambda>0$. It follows that the objective function of problem (\ref{RSCM_Lagrangian_4}) is a convex function of $\lambda>0$ and $\mu \in \mathbb{R}$ with $y=H(u,\varepsilon)-\mu$. Hence (\ref{RSCM_Lagrangian_4}) is a convex problem.
\eproof

\smallskip

 An important implication of Proposition \ref{RSCM_Characterization}  is the fact that we can characterize the function $W^{DRO}(u)$  as the solution of a \emph{finite-dimensional convex} optimization problem. The efficiency in solving program (\ref{RSCM_Duality}) strongly depends on expectation w.r.t. the nominal distribution $F$ and the properties of the convex conjugate $\phi^*$. 
 \smallskip
 
  The next corollary  formalizes  the connection between $W(u)$ and $W^{DRO}(u)$ when $\rho=0$.

%It seems that $\rho=0$ implies that $\lambda\longrightarrow\infty$. Then we should have that in Proposition  \ref{RSCM_Characterization}  $W^{DRO}(u)\longrightarrow W(u)$. Can we formalize  the following proposition? It is easy to see from the definition of the DRO-RUM but maybe we can use our characterization.
	\begin{corollary} Let Assumption \ref{RUM_conditions}  hold. Then for $\rho=0$  we get $W^{DRO}(u)= W(u)$.
\end{corollary}
\proof Let us look at problem \eqref{RSCM_Program_11}. If $\rho=0$ we get from one constraint that 
	\[
	\mathbb{E}_{F}[\phi(L(\varepsilon))] \leq 0.
	\]
	Due to the definition of $\phi$, this implies that $L(\varepsilon) =1$. Hence, the Lagrangian simplifies since the supremum over the densities becomes trivial. Let us plug  $L(\varepsilon) =1$ into Equation \eqref{RSCM_Lagrangian_11}:
	\[
	\mathcal{L}(L, \lambda, \mu)=\int_{\RR^{J+1}}H(u,\varepsilon) -\lambda \cdot \underbrace{\phi(1)}_{=0}-\mu \cdot 1] d F(\varepsilon)+\lambda \cdot 0+\mu.
	\]
	The latter is  equivalent to 
	\[
	\mathbb{E}_{F}\left[H(u,\varepsilon)  -\mu \right] + \mu = 	\mathbb{E}_{F}\left[H(u,\varepsilon) \right],
	\]
	where the last equality holds due to the linearity of expectation. We indeed recover  $W(u)$ for any distribution satisfying Assumption \ref{RUM_conditions}. \eproof
\smallskip

\subsection{A robust WDZ theorem}A fundamental aspect of RUMs is the possibility of characterizing choice probabilities under specific distributional assumptions on $\varepsilon$. Formally, and as a consequence of Assumption \ref{RUM_conditions},  the WDZ theorem establishes that the gradient of $W(u)$  yields the choice probability vector $p(u)$. In this section, we show that in the DRO-RUM, a similar result holds. In particular,  we show that $\nabla W^{DRO}(u)=p^\star(u)$ where $p^\star(u)$ corresponds to the choice probability vector generated by the optimal solution to \eqref{RSCM_Duality} approach. To establish this result, we need the following assumption.

%Therefore, we want to derive robust choice probabilities based on the robust surplus function \eqref{Robust_SCM}. In what follows we show that this robust choice probabilities can be characterized by the gradient of the robust surplus function associated with certain $\phi$-divergences. The latter are identified by an additional  assumption.

\begin{assumption}\label{phi_diffferentiable_Assumption}$\phi^*(s)$ is strictly convex and differentiable with $\phi^{*\prime}(s)\geq0$ for all $s$. 
\end{assumption}
We point out that many  $\phi$-divergence functions satisfy Assumption \ref{phi_diffferentiable_Assumption}. Table \ref{tb:div} overviews three popular $\phi$-divergences satisfying this assumption.

	\begin{table}[h]
		\begin{center}
			\begin{tabular}{ |c | c | c |c | c | c |}
				\hline
				Divergence & $\phi(t)$  &  $\phi^*(s)$ & Domain & $\phi^{*'}$  &$\phi^{*''}$ \\ \hline
				{\it Kullback-Leibler} & $t \log t$ & $e^{s-1}$ & $\RR$ & $e^{s-1}$ & $e^{s-1}$ \\ \hline
				{\it  Reverse Kullback-Leibler} & $- \log(t)$ & $ -1-\log(-s)$ & $\RR_{--}$ & $-\frac{1}{s}$  & $\frac{1}{s^2}$ \\ \hline
			{\it Hellinger Distance} & $(\sqrt{t}-1)^2$ & $\frac{s}{1-s}$ & $s < 1$ & $ \frac{1}{\left(1-s\right)^2}$ &  $-\frac{2}{\left(s-1\right)^3}$ \\ \hline
				
			\end{tabular}
			\caption{$\phi$-divergences with their convex conjugates and first and second derivatives.}\label{tb:div}
		\end{center}
	\end{table} 
\begin{comment}
	
We point out that many $\phi$-divergence functions satisfy Assumption \ref{phi_diffferentiable_Assumption}. \\  For instance, selecting $\phi(t)=t \log t$ recovers the popular  {\it Kullback-Leibler } (KL) divergence with  $\phi^{*}(s)=e^{s-1}$, which is obviously differentiable and strictly convex. \\ Similarly,  choosing $\phi(t)=- \log(t)$ yields the  {\it  Reverse Kullback-Leibler} (RKL) divergence with convex conjugate $\phi^*(s) = -1-\log(-s)$ for $s < 0$. The derivative $\phi^{*'}(s) = -\frac{1}{s}$  is clearly positive on the domain $\mathbb{R}_-$. Inspecting the second derivative, yields the strict convexity: $\phi^{*''}(s) = \frac{1}{s^2} >  0$. \\ Finally, for $\phi(t)=(\sqrt{t}-1)^2$ one obtains the  {\it Hellinger Distance} with conjugate $\phi^*(s) = \frac{s}{1-s}, \; s < 1$. Hence, $\phi^*(s)$ is differentiable with $\phi^{*'}(s) = \frac{1}{\left(1-s\right)^2} > 0\;, s<1$. Moreover, the second derivative is positive on the domain, i.\,e. $ \phi^{*''}(s) = -\frac{2}{\left(s-1\right)^3} >0$ for $s < 1$ and thus, the function is strictly convex.
\end{comment}
	
As a  direct implication of the Assumption  \ref{phi_diffferentiable_Assumption} we can establish the strict convexity and uniqueness of an optimal solution to  \eqref{RSCM_Duality}. 
\begin{lemma}\label{uniqueness_convexity_RSCM} Let  Assumptions \ref{RUM_conditions} and \ref{phi_diffferentiable_Assumption} hold.  Then program  \eqref{RSCM_Duality} is strictly convex  and has a unique optimal solution $\lambda^\star$ and $\mu^\star$.
\end{lemma}
\proof  Due to Assumption \ref{phi_diffferentiable_Assumption}, the function $\phi^*$ is strictly convex. Following  similar steps as \citet{dacorogna2008role}, it follows that  $\lambda\cdot \phi^*(\frac{s}{\lambda})$, $\lambda > 0$, is strictly convex. Further, the sum of a convex and strictly convex is strictly convex. This latter fact immediately implies strict convexity of the objective function in $\lambda$ and $\mu$. Given the strict convexity in $\lambda$ and $\mu$, it follows that program \eqref{RSCM_Duality} has a unique solution.\eproof
 \smallskip
 
 A second important implication of Assumption \ref{phi_diffferentiable_Assumption} is the possibility of characterizing the robust density associated to the optimal solution of the program \eqref{RSCM_Duality}.
 \begin{lemma}\label{optimal_mu}Let  Assumptions \ref{RUM_conditions} and \ref{phi_diffferentiable_Assumption} hold.   For a fixed $u\in \mathcal{U}$,  let $\lambda^\star>0$ and $\mu^\star\in \RR$ be the unique optimal solution to problem  (\ref{RSCM_Duality}). Then the unique robust density  $g^\star(\varepsilon)$ corresponds to:
 	\begin{equation}\label{Optimal_Density}
 		g^\star(\varepsilon)=\phi^{\ast\prime}\left({H(u,\varepsilon)-\mu^\star\over \lambda^\star}\right)f(\varepsilon)\quad\forall \varepsilon\in \RR^{J+1}.
 	\end{equation}
 \end{lemma}
 
 \proof Define $\Psi(\lambda,\mu) := \lambda\rho+\mu+\lambda\EE_{F}\left( \phi^*\left({H(u,\varepsilon)-\mu\over\lambda}\right)  \right)$. Optimizing $\Psi(\lambda,\mu)$ w.r.t $\lambda$ and $\mu$,  the first order conditions   combined with Assumption \ref{phi_diffferentiable_Assumption} yield that the optimal solution $\lambda^\star$ and $\mu^\star$  must satisfy
 \begin{eqnarray}
 	\EE_{F}\left(\phi^{\ast\prime}\left({H(u,\varepsilon)-\mu^\star\over \lambda^\star}\right)\right)&=&1\nonumber\\
 	\int_{\RR^{J+1}}\phi^{\ast\prime}\left({H(u,\varepsilon)-\mu^\star\over \lambda^\star}\right)f(\varepsilon)d\varepsilon&=1&\nonumber
 \end{eqnarray}
 Define $g^\star(\varepsilon)\triangleq\phi^{\ast\prime}\left({H(u,\varepsilon)-\mu^\star\over \lambda^\star}\right)f(\varepsilon)$. It follows  that $\int_{\RR^{J+1}}g^\star(\varepsilon)d\varepsilon=1$. Furthermore, by Assumption \ref{phi_diffferentiable_Assumption}, it follows that $g^\star(\varepsilon)\geq0$ for all $\varepsilon\in \RR^{J+1}$. Hence,  we conclude that $g^\star(\varepsilon)$ is indeed a probability density, and we call it the robust density associated with the problem (\ref{RSCM_Duality}). \eproof
 
 Some remarks are in order. First, the robust density $g^\star$ depends on the choice of the $\phi$-divergence through its conjugate $\phi^\ast$. Moreover, the robust density depends on the deterministic utility vector via $H(u,\varepsilon)$, even though the nominal distribution $F$ does not depend on $u$ due to Assumption \ref{RUM_conditions}. In addition, $g^\star$ allows us to define the robust distribution function $G^\star$, which, as we shall see, plays a key role in providing an explicit form for $W^{DRO}(u)$.   Second,  Lemma \ref{optimal_mu} establishes that   the robust density $g^{\star}(\varepsilon)$ incorporates correlation in the elements of the random vector $\varepsilon$ through the factor $\phi^{\ast\prime}((H(u,\varepsilon)-\mu^\star)/\lambda^\star)$. Thus, even though the nominal distribution $F$ may assume that $\varepsilon_0,\varepsilon_1,\ldots,\varepsilon_J$ are independent, the DRO-RUM approach introduces correlation of these terms.

\begin{example}\label{example:kld}[KL-Divergence] We now consider the case of the Kullback-Leibler divergence. In doing so, we define  $\phi$ as follows:
	\begin{equation}\label{KL_Divergence}
		\phi(t)\triangleq t\log t, \; t \geq 0
	\end{equation}
	We note that in the previous expression, $0\log 0=0$. Here 
	\begin{equation}
		\int_{\RR^J} \phi(L(\varepsilon)) dF(\varepsilon)
	\end{equation}
	defines the Kullback-Leibler divergence, denoted $D_{KL}(G \| F)$. For $\lambda>0$ the conjugate of $\lambda \phi$ is $(\lambda \phi)^*(y)=\lambda\left(e^{y / \lambda}-1\right)$. 
	From Proposition \ref{RSCM_Characterization} we know that
	\begin{equation}\label{KL_RSCM}
		W^{DRO}(u)=\inf _{\lambda \geq 0, \mu}\left\{\lambda \rho+\mu+\lambda e^{-\mu / \lambda} \mathbb{E}_F\left[e^{H(u,\varepsilon) / \lambda}\right]-\lambda\right\}
	\end{equation}
	In (\ref{KL_RSCM}) minimizing with respect to $\mu$  yields $\mu^\star=\lambda \ln \mathbb{E}_F\left[e^{H(u,\varepsilon) / \lambda}\right]$. Plugging $\mu^\star$  in (\ref{KL_RSCM}) we obtain $\lambda^\star$ as the solution to
	\begin{equation}\label{KL_RSCM_1}
		W^{DRO}(u)=\inf _{\lambda>0}\left\{\lambda \rho+\lambda \ln \mathbb{E}_F\left[e^{H(u,\varepsilon) / \lambda}\right]\right\}.
	\end{equation}
	
	It is well-known that in the case of the KL divergence (e.g., \cite{Hu_Hong_2012} and \cite{Hansen_Sargent_2001}), the ``robust'' density is given by:
	\begin{equation}\label{eq:kld.dens}
		f^{DRO}(\varepsilon)=f(\varepsilon){e^{H(u,\varepsilon)/\lambda^\star}\over \EE_{F}(e^{H(u,\varepsilon)/\lambda^\star})}
	\end{equation}
	where $f(\varepsilon)$ is the density associated to a nominal distribution, $H(u,\varepsilon)=\max_{j\in \mathcal{J}}\{ u_j+\varepsilon_j\}$ and $\lambda^\star$ is the unique optimal solution to (\ref{KL_RSCM_1}). To see how the optimal density  \eqref{eq:kld.dens} compares to the case where the nominal is a Gumbel distribution, in Figure \ref{fig:three graphs}. 
\end{example}

\begin{figure}
	\centering
	\begin{subfigure}[b]{0.3\textwidth}
		\includegraphics[width=\textwidth]{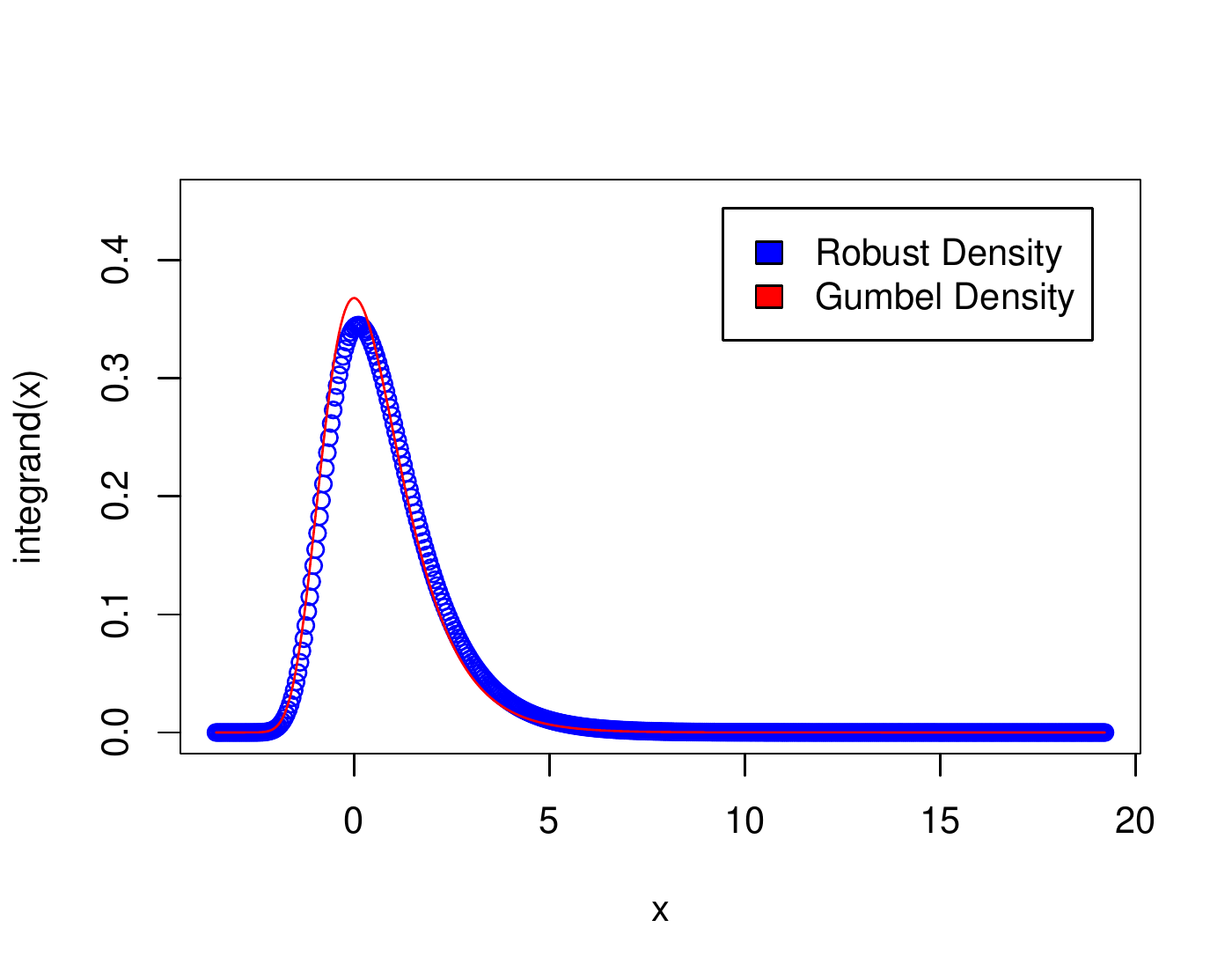}
		\caption{$\rho=0.01$}
		\label{fig:y equals x}
	\end{subfigure}
	\hfill
	\begin{subfigure}[b]{0.3\textwidth}
		\includegraphics[width=\textwidth]{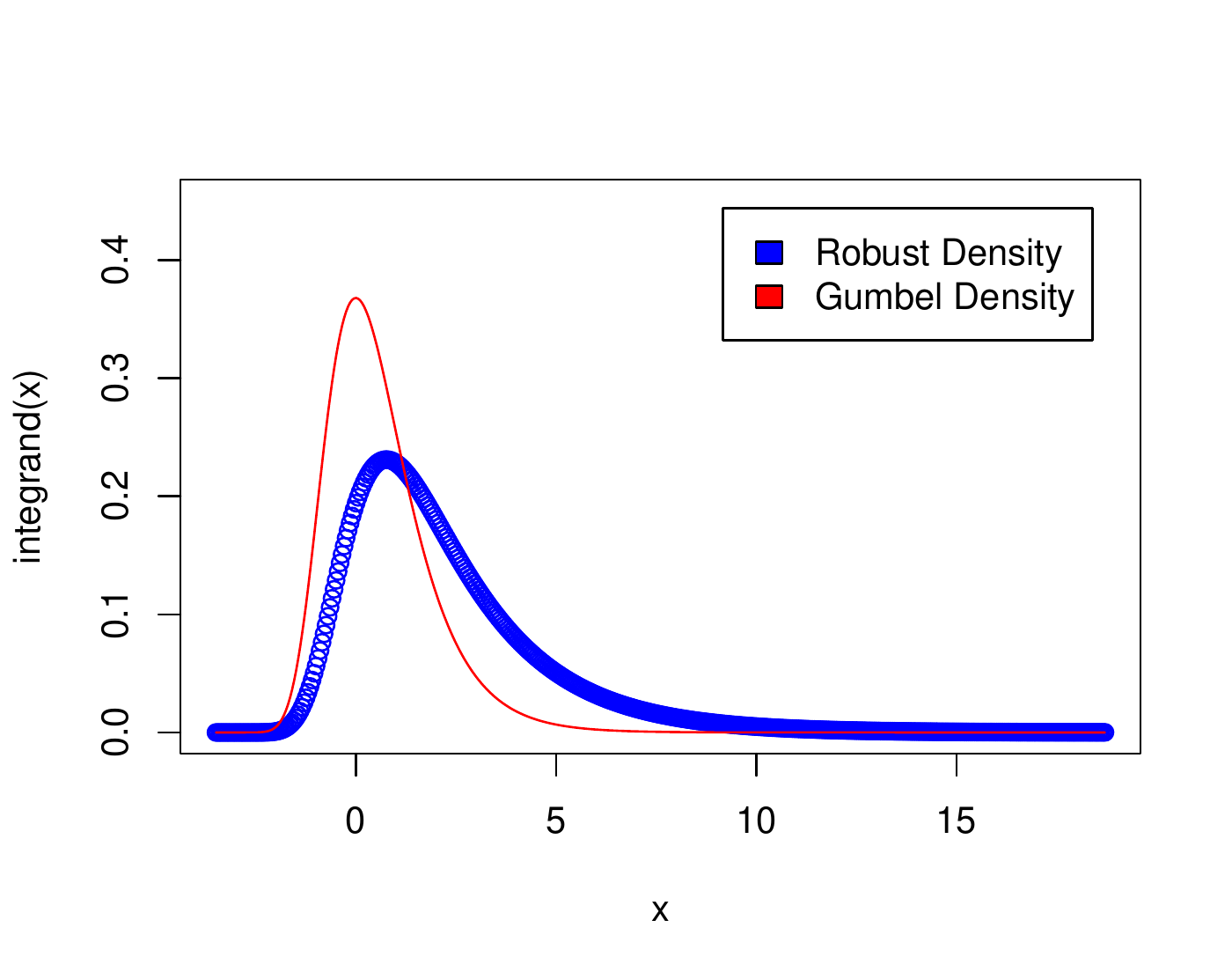}
		\caption{$\rho=0.5$}
		\label{fig:three sin x}
	\end{subfigure}
	\hfill
	\begin{subfigure}[b]{0.3\textwidth}
		\includegraphics[width=\textwidth]{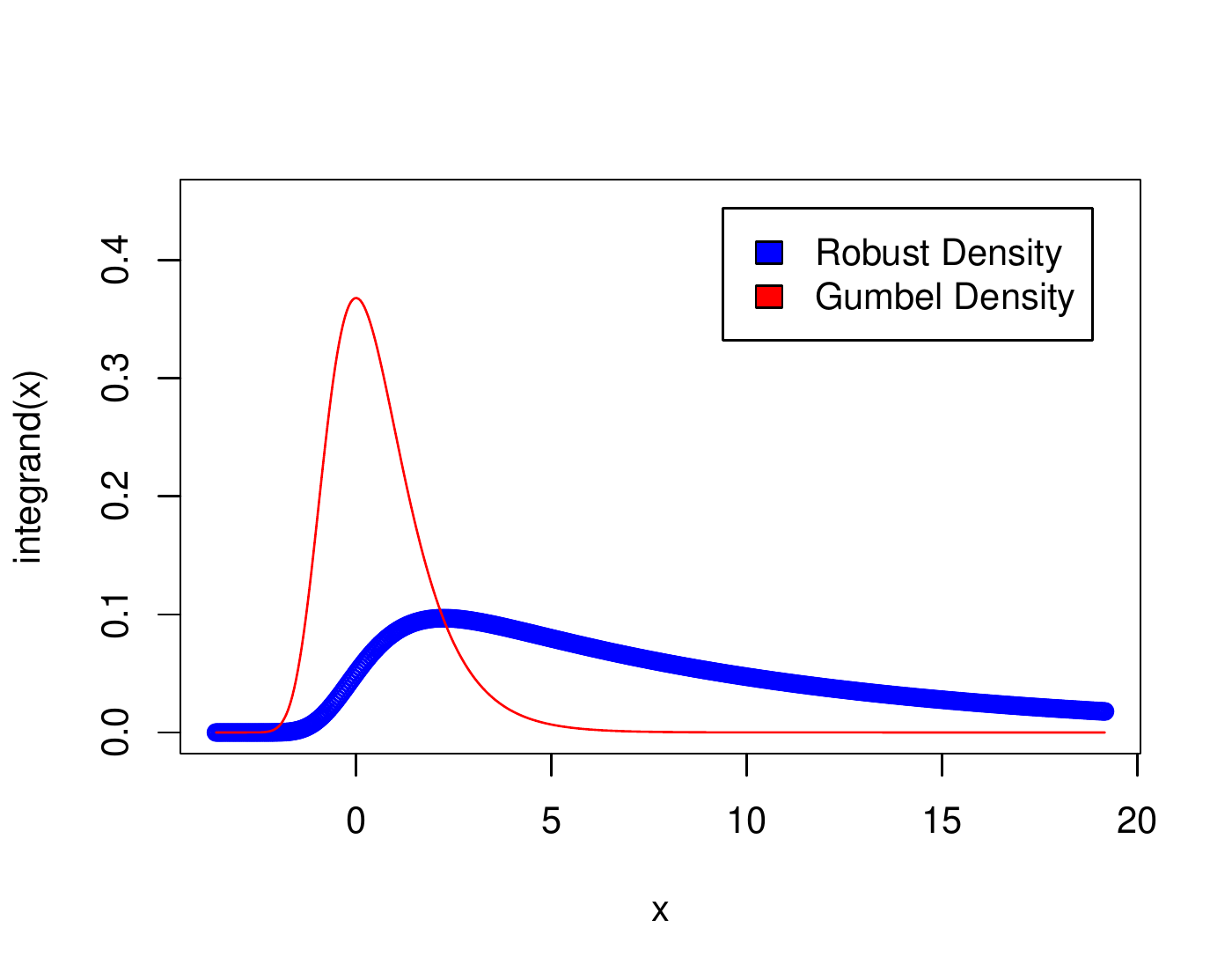}
		\caption{$\rho=3.5$}
		\label{fig:five over x}
	\end{subfigure}
	\caption{Three instances of $\rho$}
	\label{fig:three graphs}
\end{figure}
%\textcolor{red}{Include graphic in 1-dim case. Note that there, it is independent of $u$.}
%\textcolor{red}{We can plot the Gumbel distribution (density) and the adjusted density $f^{DRO}(\varepsilon)$. We can also compare this adjusted density to the mixed MNL. As the robust density depends on the level of $u$ it moght be interesting to plot the robust density for different values of $u$.}

The result in Lemma \ref{optimal_mu} enables us to characterize the choice probability vector $p^\star(u)$ similarly to the celebrated WDZ theorem.
 \begin{theorem}\label{Robust_WDZ} Let Assumptions \ref{RUM_conditions} and \ref{phi_diffferentiable_Assumption} hold. Let $\lambda^\star$ and $\mu^\star$ be the unique optimal solution to program (\ref{RSCM_Duality}), which induces  $g^\star$ and $G^\star$ as the optimal density and distribution function, respectively. Then  the following statements hold:
 	\begin{itemize}
 		\item[(i)] 	The robust social surplus corresponds to the following:
 		\begin{equation*}\label{Optimized_RSCM}
 			W^{DRO}(u)=\EE_{G^\star}\left(\max_{j\in \mathcal{J}}\{u_j+\varepsilon_j\}\right)
 		\end{equation*}
 		\item[(ii)] The choice probability vector $p^\star(u)$ is
 		$$\nabla W^{DRO}(u)=p^\star(u).$$
 	%	\begin{equation}
 		%	{\partial W^{DRO}(u)\over \partial u_j}={\partial\EE_{G^\star}\left(\max_{j^\prime\in \mathcal{J}}\{u_j+\varepsilon_j\}\right)\over \partial u_j}=	p^\star_j(u)\label{Probabilites_RSCM}
 	%	\end{equation}
 		
 	\end{itemize}
 	\end{theorem}
 \proof (i) To show the first part, let us define the function $\Psi(\cdot)$ as follows $\Psi(\lambda,\mu)\triangleq\lambda\rho+\mu+\lambda\EE_{F}\left(\phi^*\left({H(u,\varepsilon)-\mu\over \lambda}\right)\right)$. Optimizing $\Psi(\lambda,\mu)$ with respect to $\lambda$ and $\mu$ we get
 {\footnotesize
 \begin{eqnarray}
 {\partial \Psi(\lambda,
 	\mu)\over \partial \lambda}	&=&\rho+\EE_{F}\left(\phi^*\left({H(u,\varepsilon)-\mu\over \lambda}\right)\right) \\ &+& \lambda\EE_{F}\left(\phi^{\ast\prime}\left({H(u,\varepsilon)-\mu\over \lambda}\right)\left({H(u,\varepsilon)-\mu\over -\lambda^2}\right)\right)=0\label{FOC1}\nonumber\\
  {\partial \Psi(\lambda,
 	\mu)\over \partial \mu}	&=&1-\EE_{F}\left(\phi^{\ast\prime}\left({H(u,\varepsilon)-\mu\over \lambda}\right)\right)=0\nonumber\label{FOC2}
 \end{eqnarray}
}

Rearranging the  first equation, we have: 
 \begin{align*}
 &\lambda\rho+\lambda\EE_{F}\left(\phi^*\left({H(u,\varepsilon)-\mu\over \lambda}\right)\right)+\mu\EE_{F}\left(\phi^{\ast\prime}\left({H(u,\varepsilon)-\mu\over \lambda}\right)\right) \\ =&\EE_{F}\left(\phi^{\ast\prime}\left({H(u,\varepsilon)-\mu\over \lambda}\right)H(u,\varepsilon)\right).
 \end{align*}
 Similarly, in the second equation, we have:
 
 $$\EE_{F}\left(\phi^{\ast\prime}\left({H(u,\varepsilon)-\mu\over \lambda}\right)\right)=1$$

Combining both expressions we find that the optimal $\lambda^\star$  and $\mu^\star$ must satisfy:
$$\lambda^\star\rho+\lambda^\star\EE_{F}\left(\phi^*\left({H(u,\varepsilon)-\mu^\star\over \lambda^\star}\right)\right)+\mu^\star=\EE_{F}\left(\phi^{\ast\prime}\left({H(u,\varepsilon)-\mu^\star\over \lambda^\star}\right)H(u,\varepsilon)\right).$$

Using expression (\ref{Optimal_Density}) in Lemma   \ref{optimal_mu}, we obtain:
 $$\Psi(\lambda^\star,\mu^\star)=\EE_{F}\left(\phi^{\ast\prime}\left({H(u,\varepsilon)-\mu^\star\over \lambda^\star}\right)H(u,\varepsilon)\right)=\EE_{G^\star}\left(\max_{j\in \mathcal{J}}\{ u_j+\varepsilon_j\}\right). $$

Hence, we conclude that $$W^{DRO}(u)=\EE_{G^\star}\left(\max_{j\in J}\{ u_j+\varepsilon_j\}\right).$$
\smallskip
 
 (ii) To show that $\nabla W^{DRO}(u)=p^\star(u)$, we note that using the optimized value $\Psi(\lambda^\star,\mu^\star)=\lambda^\star\rho+\mu^\star+\lambda^\star\EE_{F}\left(\phi^*\left({H(u,\varepsilon)-\mu^\star\over \lambda^\star}\right)\right)$ we get:
 \begin{eqnarray}
 	{\partial W^{DRO}(u)\over \partial u_j}&=&{\partial \Psi(\lambda^\star,\mu^\star)\over \partial u_j}\nonumber\\
&=&\int_{\varepsilon\in\RR^{J+1}}\left(\phi^{\ast\prime}\left({H(u,\varepsilon)-\mu^\star\over \lambda^\star}\right){\partial H(u,\varepsilon)\over \partial u_j}\right)f(\varepsilon)d\varepsilon \nonumber\\
 	&=&\EE_{G^\star}\left({\partial H(u,\varepsilon)\over \partial u_j}\right) \nonumber\\
 	&=&p^\star_j(u).\nonumber
 \end{eqnarray}

As previous result holds for all $j\in \mathcal{J}$, we conclude that $\nabla W^{DRO}(u)=p^\star(u)$.
 \eproof
 
 Part (i) of the theorem establishes that  given the optimal solutions $\lambda^\star$ and $\mu^\star$, the surplus function $W^{DRO}(u)$ takes the familiar expected maximum form that characterizes the  RUM (see  Eq.\eqref{SSF}). The main difference between the characterization in part (i) and the surplus functions from RUM is that expression \eqref{Optimized_RSCM}  corresponds to the expectation with respect to the robust distribution $G^\star$. Part (ii)  shows that the gradient of $W^{DRO}(u)$ yields the choice probability vector $p^\star(u)$. This latter result generalizes the WDZ to environments where the nominal distribution  $F$ may be misspecified or incorrect. In other words, Theorem \ref{Robust_WDZ} shows that the DRO-RUM preserves the expected maximum form and the gradient structure of the popular RUM.

 \section{Empirical content of the DRO-RUM}\label{s4:empirical_content}
 In this section, we discuss the empirical content of the DRO-RUM.
 In particular, we show how our approach is suitable to recover the mean utility vector allowing for uncertainty about the true distribution generating $\varepsilon$.  
 \smallskip
 
 To gain some intuition, consider a situation where the choice probability vector $p$ is observed from market data. Then the analyst's goal is to find a vector $u$ that rationalizes the observed  $p$. Following \cite{Berry1994}, this problem is known as the \emph{demand inversion}. In particular,  \cite{Berry1994} shows that in the case of the MNL $u$ satisfy the following $$p_j={e^{u_j}\over 1+\sum_{j^\prime=1}^Je^{u_{j^\prime}}}\quad\mbox{for $j=1,\ldots, J.$}$$
 and 
 $$p_0={1\over 1+\sum_{j^\prime=1}^Je^{u_{j^\prime}}}$$
 
 Then  using the previous expressions, we can solve  for the mean utility vector $u$ as a function of $p$: 
 $$\log(p_j/p_0)=u_j\quad\mbox{for $j=1,\ldots,J$}.$$
 In other words, we can express $u$ in terms of the observed choice probability vector $p$.
 \smallskip
 
 We can use a similar argument to find the vector $u$ in the case of the nested logit, the random coefficient  MNL model (\cite{Berry1994, BerryLevinsohnPakes1995}),  and in the case of the inverse product differentiation logit model of \cite{Fosgerauetal2022}. For general RUMs beyond the MNL and its variants, \cite{Galichon_Salanie_2021} develops a general approach based on convex duality and mass transportation techniques. They show that for any \emph{fixed} distribution of $\varepsilon$ the mean utility vector $u$ is identified from the observed choice probability $p$.
 
 This section aims to show that the DRO-RUM  can be used to study the demand inversion problem in environments where the analyst does not know the true distribution of $\varepsilon$. Thus, our approach allows us to identify $u$ under misspecification of the distribution governing the realizations of $\varepsilon$.
 
\subsection{Robust demand inversion} Our main result uses  a distributionally robust version of the Fenchel equality for discrete choice models. In order to establish this result, we define $\mathcal{U}\triangleq \{ u\in \RR^{J+1}: u_0=0\}$. In other words, $\mathcal{U}$  is the set of mean utility vectors with the normalization $u_0=0$ for the outside option.
 Our first step is to understand the properties of the  convex conjugate of $W^{DRO}(u)$: 
 \begin{equation}\label{con.conj}
 	{W^{(DRO)}}^*(p) = \sup_{u \in \mathcal{U}} \left\{\langle u,p \rangle - W^{DRO}(u)\right\}.
 \end{equation}  
In particular, we are interested in understanding the behavior of $W^{(DRO)^*}(p)$   on its  effective domain of:
 \[
 \mbox{dom}\, {W^{(DRO)}}^* = \left\{p \in \mathbb{R}^{J+1} \,|\,  {W^{DRO}}^*(p) < \infty \right\}.
 \]

 The following lemma plays a key role in our analysis.
 \begin{lemma}\label{strict_convexity_w_dro}
 	Let Assumptions \ref{RUM_conditions} and \ref{phi_diffferentiable_Assumption} hold.  Then $W^{DRO}(u)$ is strictly convex in $u$.
 \end{lemma}
\proof For given $\lambda$ and $\mu$, for $u_1, u_2$ and  $\alpha \in \left(0,1\right)$ we have
\begin{align*}
&\lambda\EE_F\left(\phi^*\left({H(\alpha u_1 + \left(1 - \alpha\right) u_2,\varepsilon)-\mu\over \lambda}\right)\right) \\  \overset{(\star)}{\leq} \lambda&\EE_F\left(\phi^*\left( {\alpha H(u_1,\varepsilon) + \left(1 - \alpha\right) H(u_2,\varepsilon)-\mu\over \lambda}\right)\right), 
\end{align*}
where $\left(\star\right)$ holds due to the convexity of $H$ and the monotonicity of $\phi^*$ due to Assumption \ref{phi_diffferentiable_Assumption}. Exploiting the strict convexity of $\phi^*$ and the linearity and monotonicity of the expectation operator further yields:
\begin{align*}
	&\lambda\EE_F\left(\phi^*\left( {\alpha H(u_1,\varepsilon) + \left(1 - \alpha\right) H(u_2,\varepsilon)-\mu\over \lambda}\right)\right) \\ < \, &\alpha\lambda \EE_F\left(\phi^*\left(  {H(u_1,\varepsilon)-\mu\over \lambda}\right) \right) + \left(1-\alpha\right)\lambda \EE_F\left(\phi^*\left(  {H(u_2,\varepsilon)-\mu\over \lambda}\right) \right).
\end{align*}
Thus it follows that $W^{DRO}(u)$ is strictly convex in $u$. \eproof
\smallskip

The following theorem establishes the continuity and smoothness of $W^{(DRO)^*}$.

 \begin{theorem}\label{Robust_W_conjugate} Let Assumptions \ref{RUM_conditions} and \ref{phi_diffferentiable_Assumption} hold. The convex conjugate  $W^{(DRO)^*}$ is continuous on its domain $\mbox{dom}\, {W^{(DRO)}}^*$ which coincides with the probability simplex $\Delta_{J+1}$. Furthermore, $W^{(DRO)^*}$  is  continuously differentiable on $\mbox{int dom}\, W^{(DRO)^*}$.
 \end{theorem}	
 \begin{proof}
 	Let us first show that $\mbox{dom}\, {W^{(DRO)}}^* \subseteq \Delta_{J+1}$. Fix a utility vector $\bar{u}$ and take any $p \in \mathbb{R}^{J+1}$ with $\langle p,e \rangle \neq 1$. Then, using Lemma\ref{lem:surplus}(iii) we have 	
 	\begin{align*}
 	&{W^{(DRO)}}^*(p) \geq \sup_{\gamma \in \mathbb{R}} \langle p,\bar{u}+\gamma\cdot e\rangle - {W^{DRO}}(\bar{u}+\gamma\cdot e) \\ \overset{(iii)}{=} &\langle p,\bar{u}\rangle - {W^{DRO}}(\bar{u})   + \sup_{\gamma \in \mathbb{R}}\gamma\left(\langle e,p\rangle - 1\right) = \infty.
 	\end{align*}
 	Next, we take any vector $p \in \mathbb{R}^{J+1}$ with $p_i < 0$ for some $i \in \left\{0,1, \ldots, J\right\}$. By Lemma \ref{lem:surplus} (ii), it follows that
 	\[
 	{W^{(DRO)}}^*(p) \geq  \sup_{\gamma <0} \langle p, \gamma\cdot e_i\rangle - {W^{DRO}}(e_i) \overset{(ii)}{\geq} \sup_{\gamma <0} \gamma\cdot p_i,  - {W^{DRO}}(0) = \infty.
 	\]
 	Hence, it remains to prove the reverse implication, i.\,e. $\Delta_{J+1}\subseteq \mbox{dom}\, {W^{(DRO)}}^*$. Therefore, we derive an upper bound for the convex conjugate on the simplex:
 \begin{align*}
 	&\sup_{p \in \Delta_{J+1}}  	{W^{(DRO)}}^*(p) = \sup_{p \in \Delta_{J+1}} \left(\sup_{u \in \mathcal{U}} \langle p, u \rangle - W^{DRO}(u) \right) \\ = &\sup_{u \in \mathcal{U}} \left(\sup_{p \in \Delta_{J+1}} \langle p, u \rangle - W^{DRO}(u) \right).
 \end{align*}
 	We apply (iii) from Lemma \ref{lem:surplus} which yields
 	\[
 	\sup_{u \in \mathcal{U}} \left(\sup_{p \in \Delta_{J+1}} \langle p, u \rangle - W^{DRO}(u) \right) = \sup_{u \in \mathcal{U}}  \left(\max_{i \in\mathcal{J}} u_i - W^{DRO}(u) \right) \leq -\min_{i\in\mathcal{ J}} \mathbb{E}_{F}\left[\varepsilon_i\right].
 	\]
 	Thus, the domain coincides with the simplex.
 	For the continuity, we first observe that ${W^{(DRO)}}^*$ is convex, and hence it is continuous on the relative interior of its domain. The Gale-Klee-Rockafellar theorem provides upper semi-continuity of ${W^{(DRO)}}^*$ if the domain is polyhedral, which it is \citep{Rockafellar1970}. Furthermore, convex conjugates are always lower semi-continuous, and hence continuity follows. In order to establish that $W^{(DRO)^*}$  is continuously differentiable, we note that Lemma \ref{strict_convexity_w_dro}  shows that  $W^{(DRO)^*}$   is strictly convex in $u$. Then by \citet[Thm. 4.1.1]{Hiriart_Urruty_Lemarechal_1993} we know that the strict convexity of $W^{DRO}(u)$ implies that $W^{(DRO)^*}(p)$ is continuously differentiable on $\mbox{int} \left( \mbox{dom}\,W^{(DRO)^*}\right)$.	\eproof
 \end{proof}
 
 The previous result is key in our goal of identifying the mean utilities. To see this we note that thanks  to Theorem \ref{Robust_WDZ} we know that  for alternative $j\in \mathcal{J}$ :
 $$p_j={\partial W^{DRO}(u)\over \partial u_j}$$
 Furthermore, from Theorem \ref{Robust_W_conjugate} we get:
 $$u_j={\partial W^{(DRO)^*}(p)\over \partial p_j}$$
 where $u$  achieves the maximum in (\ref{con.conj}). Then by Fenchel's duality theorem, we know that these two conditions are equivalent. Then, given the  \emph{robust} distribution $G^\star$,  we conclude that $u$ is identified from $p$. In other words, we can find a vector $u$ that rationalizes the observed choice probability vector $p$.
 \smallskip
 
 The following result establishes the empirical content of the DRO-RUM.

\begin{theorem}\label{Robust_Fenchel} Let Assumptions \ref{RUM_conditions} and \ref{phi_diffferentiable_Assumption} hold. Then, the following statements are equivalent:
	\begin{itemize}
		\item[(i)] The choice probability vector $p\in \Delta_{J+1}$ satisfies:
		\begin{equation}\label{Fenchel1}
		p=\nabla W^{DRO}(u).
		\end{equation}
		\item[(ii)] The deterministic utility vector $u\in \mathcal{U}$ satisfies:
		\begin{equation}
		u=\nabla W^{(DRO)^*}(p).
		\end{equation}
	\item[(iii)] $(u^\star,\lambda^\star,\mu^\star)$ is the  unique solution to the strictly convex optimization problem:
	\begin{equation}\label{Demand_Inversion_equivalence}
	-W^{(DRO)^\ast}(p)=	\inf_{u\in\mathcal{U},\lambda\in \RR_+,\mu\in \RR}\left\{\lambda\rho+\mu+\lambda\EE_F\left(\phi^*\left({H(u,\varepsilon)-\mu\over \lambda}\right)\right) -\langle p,u\rangle \right\}.
	\end{equation}
	\end{itemize}
	\end{theorem}
\proof The equivalence of parts (i) and (ii) follows from Theorem \ref{Robust_W_conjugate},  which allows us to invoke  Fenchel equality to conclude the result. To show part (iii), let us look at $W^{(DRO)^*}(p)$. By definition, we know that 
$$W^{(DRO)^*}(p)=\sup_{u\in \mathcal{U}}\{\langle p,u\rangle -W^{DRO}(u) \}.$$
Proposition \ref{RSCM_Characterization} implies that the previous expression corresponds to
$$W^{(DRO)^*}(p)=\sup_{u\in\mathcal{U}}\left\{\langle p,u\rangle -\inf_{\lambda>0,\mu\in \RR}\left\{\lambda\rho+\mu+\lambda\EE_F\left(\phi^*\left({H(u,\varepsilon)-\mu\over \lambda}\right)\right)\right\} \right\}.$$
Equivalently, we have:
$$W^{(DRO)^*}(p)=- \inf_{u\in\mathcal{U},\lambda>0,\mu\in \RR}\left\{\lambda\rho+\mu+\lambda\EE_F\left(\phi^*\left({H(u,\varepsilon)-\mu\over \lambda}\right)\right)-\langle p,u\rangle\right\}. $$
Thus, we get:
$$-W^{(DRO)^*}(p)=\inf_{u\in\mathcal{U},\lambda>0,\mu\in\RR}\left\{\lambda\rho+\mu+\lambda\EE_F\left(\phi^*\left({H(u,\varepsilon)-\mu\over \lambda}\right)\right)-\langle p,u\rangle\right\}.$$
%which is a convex program in $u,\lambda$, and $\mu$.\\
Combining Lemmas \ref{uniqueness_convexity_RSCM} and \ref{strict_convexity_w_dro}  we get that the (\ref{Demand_Inversion_equivalence}) is strictly convex in $u,\lambda,$ and $\mu$. As a consequence, there exists a unique solution to the problem (\ref{Demand_Inversion_equivalence}).\eproof
\smallskip

As we discussed in the introduction of this section, for a fixed distribution of $\varepsilon$, parts (i) and (ii) have been established in the \cite{Galichon_Salanie_2021}. Our result differs from theirs in a fundamental aspect; we achieve the identification of the mean utility vector $u$, relaxing the assumption that the distribution of $\varepsilon$ is known. In other words, our result allows for nonparametric identification of $u$ under  (potential) misspecification of the shock distribution. Similarly, our result relates to dynamic discrete choice models' ``inversion'' approach.\footnote{ It is worth remarking that \cite{Chiong_Galichon_Shum_2016} apply \cite{Galichon_Salanie_2021}'s approach  widely used in dynamic discrete-choice models.}   For instance the papers by \cite{Hotz_Miller_1993} and \cite{Arcidiacono_Miller_2011}  establish that the mean utility vector $u$ can be recovered   as  $\nabla^{-1} W(p)=u$. Their approach only applies to the case of the MNL and GEV models. By exploiting convex optimization techniques, \cite{Fosgerauetal2021109911} extends \cite{Hotz_Miller_1993} and \cite{Arcidiacono_Miller_2011}'s inversion approach to models far beyond the GEV class. Similarly,  \cite{Li2018} considers a convex minimization algorithm to solve the demand inversion problem. He illustrates his method in the case of both the \cite{BerryLevinsohnPakes1995} random coefficient logit demand model and the \cite{BerryPakes2007} pure characteristics model. However,  \cite{Fosgerauetal2021109911} and \cite{Li2018}'s results only apply under the assumption that the distribution of $\varepsilon$ is known. In contrast, part (iii) establishes that given a choice probability vector $p$, we can identify the mean utility vector $u$ as the unique solution of the strictly convex optimization program (\ref{Demand_Inversion_equivalence}). This latter characterization captures the role of misspecification through the value of the Lagrange multipliers $\lambda^\star$ and $\mu^\star$. Thus, Theorem \ref{Robust_Fenchel} provides a distributionally robust nonparametric identification result.
\subsection{A robust random coefficient model} To see how Theorem \ref{Robust_Fenchel} can be applied, we analyze the random coefficient model assuming that the $\phi$-divergence corresponds to the  Kullback-Leibler distance. Following \cite{BerryLevinsohnPakes1995} and \cite{Galichon_Salanie_2021}, we consider a random coefficient model with  $\varepsilon=Z e+ T \eta$, where $e$ is a random vector on $\mathbb{R}^k$ with distribution $F_e$ ,  $Z$ is a $\left|\mathcal{J}\right| \times k$ matrix, $T>0$ is a scalar parameter, and $\eta$  is a vector of $|\mathcal{J}|$ Gumbel random variables, whose distribution function is $F_\eta$. Assume that $e$ and $\eta$ are statistically independent. 
Fixing the distributions $F_e$ and $F_\eta$, we  can use the iterated expectation, combined with the independence of $e$ and $\eta$ (\cite[Eqs. B.6-B.7]{Galichon_Salanie_2021}) we get that  

\begin{eqnarray}
	W(u)&=&\EE_{F_e}\left(\EE_{\eta} \left(\max_{j\in \mathcal{J}}\{u_j+(Ze)_{j}+T\eta_j\}\right)|e\right),\nonumber\\
	&=&\EE_{F_e}\left(W(u+Ze)\right),\nonumber
\end{eqnarray}
where $W(u+Ze)=\int_{\RR^{J+1}}\max_{j\in \mathcal{J}}\{u_j+(Ze)_{j}+T\eta_j\}f_{\eta}d\eta$. Using the fact that $\eta$ follows a Gumbel distribution, we find that
$$W(u+Ze)=T\log\left(\sum_{j\in \mathcal{J}}{e^{u_j+(Ze)_j\over T}}\right).$$
Let us assume that  $F_e$ approximates the true distribution generating $e$. Then we can define $W^{DRO}(u)$ as follows:
$$W^{DRO}(u)=\sup_{G_e\in\mathcal{M}_{\phi}(F_e)}\EE_{G_e}\left(T\log\left(\sum_{j\in\mathcal{J}}e^{{u_j+(Ze)_j\over T}}\right)\right)$$

To apply Theorem \ref{Robust_Fenchel}, we note that $H(u,\varepsilon)=T\log\left(\sum_{j\in\mathcal{J}}e^{{u_j+(Ze)_j\over T}}\right)$. Then, using the Kullback-Leibler distance, we have that for an observable choice probability vector $p$, the identified mean utility vector $u^\star$  corresponds to the solution of the following program:
\begin{equation}\label{Robust_RCM}
	-W^{(DRO)^\ast}(p)=\inf_{u\in \mathcal{U},\lambda>0}\left\{ \rho\lambda+\lambda\ln \EE_{F_e}\|e^{u+Ze}\|_{T^{-1}}-\langle p ,u\rangle\right\},
\end{equation}
  where $\|e^{u+Ze}\|_{T^{-1}}\triangleq \left(\sum_{j\in\mathcal{J}} e^{{u_j+(Ze)_j\over T}}\right)^T.$ 
  
  The program (\ref{Robust_RCM}) allows us to identify the mean utility vector enabling some degree of misspecification in the distribution of $e$. It is worth remarking that program (\ref{Robust_RCM}) is  fairly tractable, so we can use traditional stochastic programming algorithms to find its unique solution.
 \section{Numerical Experiments}\label{s5:numerical_experiments}
 In this section, we discuss numerical simulations of our approach. We compare the DRO-RUM with the MNL and MNP models.\footnote{ We recall that the MNP assumes that the error terms follow a  normal distribution with a specific variance-covariance matrix. } 
 \smallskip  
 
 Our main goal is to analyze the effect of the robustness index $\rho$ on the choice probabilities. We consider a scenario with four alternatives where $\mathcal{J}=\{0,1,2,3\}$. Our first parametrization of the utility vector $u$ is  $u=\left(0,1,2,2.1\right)^T$. Based on this specification, we proceed to calculate the choice probabilities. In the case of the MNL, the choice probabilities are computed via Eq. \eqref{MNL_Choice} , where the scale parameter equals one ( $\eta=1$). In addition, we set the location parameter of each Gumbel error is assumed to zero.  %We call this model {\it MNL}. 
 For the MNP, we consider two different parametrizations for the variance-covariance matrix  of the random error vectors;  $\mathcal{N}(0, \Sigma_1)$ and $\mathcal{N}(0, \Sigma_2)$ where
 \[
 \Sigma_1 = \left(\begin{array}{ccccc}
  1 & 0  & 0 &  0 \\
 0 & 1 & 0 & 0 \\
0   & 0 & 1 & 0  \\
 0 & 0  & 0 &  1\\
 \end{array} \right), \qquad  \Sigma_2 = \left(\begin{array}{ccccc}
 2 & -0.5  & 0.5 &  1.3 \\
 -0.5 & 2 & 0 & 0.15 \\
 0.5   & 0 & 2 & 1  \\
 1.3 & 0.15  & 1 &  2\\
 \end{array} \right).
 \]
We call the latter model MNP-dep and the former MNP-indep, as the random errors $\varepsilon^{(i)}$, $i=1, \ldots, 4$ are independent in the former model. We use 10,000,000 draws from the error vectors to stabilize the simulations to simulate the choice probabilities. \\ For the DRO-RUMs, we choose the Kulback-Leibler- divergence case presented in Example \ref{example:kld}. We assume that the error terms of the nominal distribution are iid Gumbel distributed with location parameter zero and scale parameter one. This yields a way to examine the behavior and numerical stability of the DRO-RUM, and the impact of $\rho$ on the choice probabilities. 
\smallskip

The robust choice probabilities are simulated similarly to the MNP models. However, for the case of DRO-RUM we have to generate samples from the distribution defined by the density \ref{eq:kld.dens}. 
First, the optimal $\lambda^\star$ in \eqref{KL_RSCM_1} as well as $\mathbb{E}_{F}\left(e^{H(u,\varepsilon)/\lambda^\star}\right)$ are estimated using $50,000,000$ simulations from $4$ iid Gumbel distributions. Based on the optimized parameters a higher dimensional acceptance-rejection algorithm provides an efficient sampling method. For performance, the code was written in Julia.\footnote{ The code can be found on Github under \texttt{https://github.com/rubsc/rejection\_DRO\_RUM}.}
\smallskip

We present the results in Table \ref{tab:prob}
\begin{table}[H]
	\begin{center}
		\begin{tabular}{c | c | c | c| c|}
		
			& Alternative $1$ & Alternative $2$ & Alternative $3$ & Alternative $4$ \\ \hline
			 MNL & $5.1885 \% $ & $ 14.1037 \%$ & $38.3379 \%$ & $42.3699 \%$ \\ \hline
			MNP-indep & $1.6243 \% $ & $ 10.2996 \%$ & $41.443 \%$ & $46.6331 \%$  \\ \hline
			MNP-dep& $1.7877 \% $ & $ 19.359 \%$ & $37.9993 \%$ & $40.854 \%$ \\ \hline
			$\rho = 0.1$ & $9.2391 \% $ & $ 18.4783 \%$ & $33.695 \%$ & $38.587 \%$ \\ \hline
			$\rho = 0.7$ & $13.2132 \% $ & $ 19.7823 \%$ & $31.6066 \%$ & $35.3979 \%$ \\ \hline
			$\rho = 1.3$ & $15.4725 \% $ & $ 21.5501 \%$ & $30.5269 \%$ & $32.4505 \%$ \\ \hline
			$\rho = 2.2$ & $18.1501 \% $ & $ 21.46 \%$ & $29.7536 \%$ & $30.6363 \%$ \\ \hline
			$\rho = 4.3$ & $21.5843 \% $ & $ 23.2006 \%$ & $27.1235 \%$ & $28.0916 \%$ 
		\\ 
		\end{tabular}
		\caption{Choice Probabilities for utility vector $u=\left(0,1,2,2.1\right)^T$}.\label{tab:prob}
	\end{center}
\end{table} 
 In the previous table, the first row displays the choice probability for the MNL. The second and third rows show the choice probabilities for the MNP-indep and MNP-dep. The fourth row shows the behavior of the DRO-RUM when $\rho=0.1$. For this parametrization, the DRO-RUM yields choice probabilities that are similar (not equal)  to the ones displayed by the MNP-dep. Similar behavior is observed for the case of $\rho=0.7$. 
 \smallskip
 
 Rows six to eight show the behavior of the DRO-RUM as we increase  $\rho$. As expected, as the value of  $\rho$ increases, the choice probabilities look similar to the uniform choice between alternatives. In particular, for the case of $\rho=4.3$ we note that DRO-RUM assigns probabilities similar to the uniform case. Intuitively, a large $\rho$ , represents a situation where the analyst is highly uncertain about the true distribution. Thus, her behavior is overly cautious and considers a large set of possible (and feasible) distributions. Hence,  when $\rho\longrightarrow\infty$, the analyst's best choice is to guess uniform probabilities.
 \smallskip
 
 Similarly, from the  DM's perspective, large values of $\rho$ indicate a cautious and flexible choice of the error term. Consequently, the random error term might follow a distribution that completely counteracts the deterministic utilities' effects and guarantees the same overall random utility for every alternative. Indeed, the robust surplus function \eqref{KL_RSCM} is strongly increasing with a larger index of robustness as shown in Figure \ref{fig:surplus}, where we plot the surplus function evaluated at $u=\left(0,1,2,2.1\right)^T$ for different values of $\rho$.

\begin{figure}[H]
	\includegraphics[width=\linewidth]{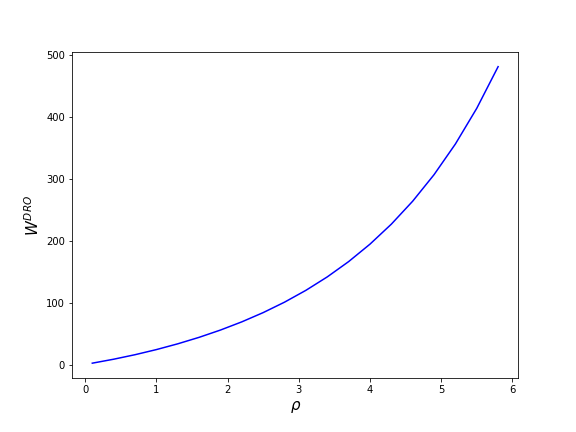}
	\caption{Robust surplus $W^{DRO}$ for  the utility vector $u$ and different values of robustness index $\rho$.}
	\label{fig:surplus}
\end{figure}
%\textcolor{red}{Discussion of the results.} 

\smallskip

A well-known pitfall of the  MNL model is that it satisfies the independence of irrelevant alternatives (IIA)  property. The IIA property establishes that the ratio between the probabilities of any two alternatives only depends on the differences between the utilities of these two alternatives. This property follows directly via formula \eqref{MNL_Choice}.  A direct implication of this fact is that when the deterministic utility of one alternative changes, the choice probabilities change proportionally so that the probability ratio between alternatives remains constant. 
In contrast, the DRO-RUM incorporates some dependence structure into the  MNL.\footnote{We recall that we are assuming that the nominal distribution is Gumbel.} Hence, it is interesting to simulate choice probabilities for a slight change in the deterministic utility vector. In Table \ref{tab:prob2}, we summarize the choice the probabilities for the alternatives with utility vector $\tilde{u}=\left(0,1,2,2.2\right)^T$.

\begin{table}[H]
	\begin{center}
		\begin{tabular}{c | c | c | c| c|}
			
			& Alternative $1$ & Alternative $2$ & Alternative $3$ & Alternative $4$ \\ \hline
			MNL & $4.9671 \% $ & $ 13.5021 \%$ & $36.7024 \%$ & $44.8284 \%$ \\ \hline
			MNP-indep & $1.4758 \% $ & $ 9.5821 \%$ & $39.2676 \%$ & $49.6745 \%$  \\ \hline
			MNP-dep & $1.5544 \% $ & $ 18.5851 \%$ & $36.0749 \%$ & $43.7856 \%$ \\ \hline
			$\rho = 0.1$ & $6.8788 \% $ & $ 15.0473 \%$ & $36.5434 \%$ & $41.5305 \%$ \\ \hline
			$\rho = 0.7$ & $13.2076 \% $ & $ 20.2437 \%$ & $30.8569 \%$ & $35.6918 \%$ \\ \hline
			$\rho = 1.3$ & $15.2938 \% $ & $ 21.5069 \%$ & $30.7281 \%$ & $32.4712 \%$ \\ \hline
			$\rho = 2.2$ & $17.9704 \% $ & $ 21.4406 \%$ & $29.5254 \%$ & $31.0636 \%$ \\ \hline
			$\rho = 4.3$ & $21.5509 \% $ & $ 23.264 \%$ & $27.4421 \%$ & $27.743 \%$ 
			\\ 
		\end{tabular}
		\caption{Choice Probabilities for utility vector $\tilde{u}=\left(0,1,2,2.2\right)^T$}\label{tab:prob2}
	\end{center}
\end{table} 

The violation of IIA, is visualized in Figure \ref{fig:prob}. Note that in the MNL, the decrease in choosing alternative $4$  evenly increases the probability of choosing one of the alternatives $1-3$, indicated by the dotted line. At the same time, the substitution patterns for the robust models are way more flexible. 

\begin{figure}[H]
	\includegraphics[width=\linewidth]{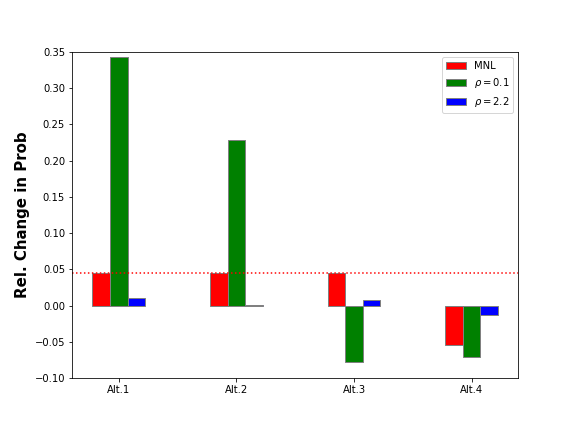}
	\caption{Relative change in probabilities if deterministic utility vector changes from $\tilde{u}$ to $u$.}
	\label{fig:prob}
\end{figure}

\section{Final remarks}\label{s6:final_remarks}
In this paper, we have introduced the DRO-RUM, which allows the shock distribution to be unknown or misspecified. We have shown that the DRO-RUM preserves the tractability and convex structure of the traditional RUM. Furthermore, we characterized the empirical content of the DRO-RUM, establishing that for an observed choice probability vector, there exists a unique mean utility vector that rationalizes the observed behavior in terms of a DRO-RUM. Finally, we showed the stability and numerical properties of our approach.
\smallskip

Several extensions are possible. First,  a natural question is about the econometric performance of the DRO-RUM using market data. This is in particular interesting, as our our approach provides the analyst with a rich class of various models. In fact, different models can be created by simply choosing a different $\phi$ - divergence and/or nominal distribution. In this context, it is also interesting to examine the impact of the robustness parameter $\rho$. 
\smallskip

 Second, results of RUM could be analyzed in the framework of robustness. For instance, the results in this paper could help study two-sided matching markets with transferable utility. Similarly, our results can help study robust identification in dynamic discrete choice models.
 
\smallskip

The algorithmic aspects of the DRO-RUM could be analyzed. Recently for example, a new family of prox-functions on the probability simplex based on discrete choice models has been introduced by \citet{muller2022discrete}. Hence, it is interesting to see if prox-functions can be generated from the DRO-RUM.

\smallskip

Additionally, theoretical extensions of the distributionally robust approach are conceivable. A natural way to do this is to rely on different statistical distance  concepts, e.\, g. Wasserstein-Distance, and analyze their tractability. Moreover, the properties of such other robust models could be compared with the DRO-RUM.

\begin{comment}
 \subsection{Optimizing tail perfomance}
Now we focus in the important case of   the  Cressie-Read   divergence, which is denoted as $\phi^k_{cr}$ (\cite{Cressie_Read_1984}). We want to connect Theorem \ref{Robust_Fenchel} with an explicit  expression for the convex conjugate $\phi_{cr}^{*k}$.  From  Table \ref{table:examples} we know that  the Cressie-Read family is parameterized by $k \in$ $(-\infty, \infty) \backslash\{0,1\}$ and $k_{*}=\frac{k}{k-1}$ with
\begin{equation}\label{Cressie_Read_Divergences}
	\phi_{cr}^k(t)\triangleq\frac{t^{k}-k t+k-1}{k(k-1)} \quad\text {and}\quad \phi_{cr}^{*k}(s)\triangleq\frac{1}{k}\left[((k-1) s+1)_{+}^{k_{*}}-1\right] \text {. }
\end{equation}

For $t<0$  we let $\phi_{cr}^k(t)=+\infty$  and we define $\phi_{cr}^1$ and $\phi_{cr}^{0}$ as their respective limits as $k \rightarrow 0,1$. The family of divergences (\ref{Cressie_Read_Divergences}) includes $\chi^{2}$-divergence $(k=2)$, empirical likelihood $\phi_{cr}^0(t)=$ $-\log t+t-1$, and KL-divergence $\phi_{cr}^{1}(t)=t \log t-t+1$ as particular cases.
\smallskip

\begin{proposition}\label{Simplified_Dual1}Let Assumption \ref{RUM_conditions} hold and  let  $k \in(1, \infty), k_{*}=k /(k-1)$, any $\rho>0$, and $m_{k}(\rho):=(1+k(k-1) \rho)^{\frac{1}{k}}$. Then for all $u\in \mathcal{U}$ we have:
	\begin{equation}\label{Simplified_Dual}
		-W^{(DRO)^*}(p)=\inf _{u\in\mathcal{U}\mu \in \mathbb{R}}\left\{m_{k}(\rho) \mathbb{E}_{F}\left[(H(u , \varepsilon)-\mu)_{+}^{k_{*}}\right]^{\frac{1}{k_{*}}}-\mu\right\}.
	\end{equation}
\end{proposition}
\end{comment}

\bibliographystyle{nourl.bst} % We choose the "plain" reference style
\bibliography{Demand_refs}
\newpage
\end{document}